\documentclass{article}

\usepackage{amsfonts,amssymb,amsmath,a4wide}
\usepackage{epsfig,rotating,framed,enumerate,listings,nicefrac}
\usepackage{paralist}
\usepackage{marginnote}
\usepackage{authblk}

\newtheorem{theorem}{Theorem}[section]
\newtheorem{example}[theorem]{Example}
\newtheorem{lemma}[theorem]{Lemma}

\newtheorem{corollary}[theorem]{Corollary}
\newtheorem{definition}[theorem]{Definition}

\newtheorem{remark}[theorem]{Remark}
\newtheorem{observation}[theorem]{Observation}
\newenvironment{proof}{\noindent{\bf Proof~}}{\null\hfill $\Box$\par\medskip}

\newcommand{\IN}{\mathbb{N}}

\newcommand{\bigo}{\ensuremath{\mathcal{O}}}

\newcommand{\fpt} {\mbox{FPT}}
\newcommand{\xp} {\mbox{XP}}
\newcommand{\w} {\mbox{W}}
\newcommand{\pfpt} {\mbox{PFPT}}
\newcommand{\val} {\text{val}}

\newcommand{\sizevar} {\text{size-var}}
\newcommand{\pvar} {\text{profit-var}}
\newcommand{\p} {\mbox{P}}
\newcommand{\np} {\mbox{NP}}

\newcommand{\Q}{\mathbb{Q}}
\newcommand{\N}{\mathbb{N}}
\newcommand{\Z}{\mathbb{Z}}
\newcommand{\sgn} {\operatorname{sign}}

\newcommand{\gansfuss}[1]{\mbox{``}{#1}\mbox{''}}

\newenvironment{desctight}
  {\begin{list}{}{
\setlength\labelwidth{-5pt}
        \setlength{\itemsep}{0.5pt}
        \setlength{\parsep}{0pt}
        \setlength\itemindent{-\leftmargin}
        
}}
    {\end{list}}

\begin{document}

\title{Knapsack Problems: A Parameterized Point of View\thanks{Short versions of this paper appeared in Proceedings of the {\em International Conference on Operations Research} (OR 2014) \cite{GRY16} and (OR 2015) \cite{GRY16a}.}}

\author[1]{Carolin Albrecht}
\author[1]{Frank Gurski}
\author[2]{Jochen Rethmann}
\author[1]{Eda Yilmaz}

\affil[1]{\small University of  D\"usseldorf,
Institute of Computer Science, Algorithmics for Hard Problems Group,\newline 40225 D\"usseldorf, Germany}

\affil[2]{\small Niederrhein University of Applied Sciences,
Faculty of Electrical Engineering and Computer Science, 47805 Krefeld,
Germany}

\maketitle

\begin{abstract}
The knapsack problem (KP) is a very famous NP-hard problem in combinatorial optimization. 
Also its generalization to multiple dimensions named d-dimensional knapsack problem (d-KP) and
to multiple knapsacks named multiple knapsack problem (MKP) are well known problems.
Since KP, d-KP, and MKP are integer-valued problems defined on inputs of various informations, 
we study the fixed-parameter tractability of these problems. 
The idea behind fixed-parameter tractability is to split the 
complexity into two parts - one part that depends purely on the size of the input, and one part 
that depends on some parameter of the problem that tends to be small in practice.  Further we 
consider the closely related question, whether the sizes and the values can be reduced, such that 
their bit-length is bounded polynomially or even constantly in a given parameter, i.e. 
the existence of kernelizations is studied.
We discuss the following parameters: the number of items, the threshold value for the profit, the sizes, 
the profits, the number $d$ of dimensions, and the number $m$ of knapsacks.
We also consider the connection of parameterized knapsack problems to linear programming, 
approximation, and pseudo-polynomial algorithms.

\end{abstract}

\section{Introduction}

The knapsack problem is one of the famous tasks in combinatorial optimization.\footnote{The 
knapsack problem obtained its name 
by the following well-known example.
Suppose a hitchhiker needs to fill a knapsack for a
trip. He can choose between $n$ items  and each of
them has a profit of $p_j$ measuring the usefulness
of this item during the trip and a size $s_j$.
A natural constraint is that the total size of
all selected items must not exceed the capacity $c$ of the
knapsack. The aim of the hitchhiker
is to select a subset of items while maximizing the
overall profit under the capacity constraint.} 
In the knapsack problem (KP) we are given a set $A$ of $n$ items. 
Every item $j$ has a profit $p_j$ and a size $s_j$. 
Further there is a capacity $c$ of the knapsack. The task is to choose a subset $A'$ of $A$, 
such that the total profit of $A'$ is maximized and the total size of $A'$ is at most $c$.
Within the d-dimensional knapsack problem (d-KP) a set A of $n$ items and a number $d$ 
of dimensions is given. Every item $j$ has a profit $p_j$ and for dimension $i$ the size $s_{i,j}$. 
Further for every dimension $i$ there is a capacity $c_i$. The goal is to find a subset $A'$ of $A$, 
such that the total profit of $A'$ is maximized and for every dimension $i$ the total size of $A'$ 
is at most the capacity $c_i$.
Further we consider the multiple knapsack problem (MKP) where beside $n$ items a number $m$ of 
knapsacks is given. Every item $j$ has a profit $p_j$ and a size $s_j$ and each knapsack $i$ has a 
capacity $c_i$. The task is to choose $m$ disjoint subsets of $A$ such that the total profit of the 
selected items is maximized and each subset can be assigned to a different knapsack $i$ without exceeding 
its capacity $c_i$ by the sizes of the selected items.
Surveys on the knapsack problem and several of its variants can be found in books
by Kellerer et al. \cite{KPP10} and by Martello et al. \cite{MT90}.

The knapsack problem arises in resource allocation where there are 
financial constraints, e.g. capital budgeting.
Capital budgeting problems\footnote{In the field of finance
problems defined on projects with a \gansfuss{take it or leave it} opportunity
are denoted as {\em capital budgeting problems} (cf. Section 11.2 of \cite{CT13}).} 
have been introduced in the 1950s by Lorie and Savage \cite{LS55} and also by
Manne and Markowitz \cite{MM57}
and a survey can be found in \cite{Wei66}.

From a computational point of view the knapsack problem is intractable \cite{GJ79}.
This motivates us to consider the fixed-parameter tractability and
the existence of kernelizations of knapsack problems.
Beside the standard parameter $k$, i.e. the threshold value for the profit
in the decision version of these problems, and the number of items, 
knapsack problems offer a large number of interesting parameters.
Among these are the sizes, 
the profits, the number of different sizes,  the number of different profits, 
the number $d$ of dimensions, the number $m$ of knapsacks,
and combined parameters on these. Such parameters were considered 
for fixed-parameter tractability of the subset sum problem, which can be regarded as a
special case of the knapsack problem, in \cite{FGR10} and 
in the field of kernelizaton in \cite{EKMR15}.

This paper is organized as follows.
In Section \ref{sec-pre},  we give preliminaries on  fixed-parameter tractability and
kernelizations, which are two equivalent concepts within parameterized complexity theory. 
We give a characterization for the special case of polynomial fixed-parameter
tractability, which in the case of integer-valued problems is a super-class
of the set of problems allowing polynomial time algorithms. We show that a
parameterized problem can be solved by a polynomial fpt-algorithm if and
only if it is decidable and has a kernel of constant size. This
implies a tool to show kernels of several knapsack problems. 
Further we cite a  useful theorem for finding 
kernels of knapsack problems with respect to parameter $n$  
by compressing large integer values to smaller ones.
We also give results on the connection between the existence
of parameterized algorithms, approximation algorithms, 
and pseudo-polynomial algorithms.
In Section \ref{sec-kp}, we consider the knapsack problem. We apply
known results as well as our characterizations to show fixed-parameter tractability and
the existence of kernelizations.
In Section \ref{sec-mkp}, we look at the d-dimensional knapsack problem.
We show that the problem is not pseudo-polynomial in general by
a pseudo-polynomial reduction from {\sc Independent Set}, but pseudo-polynomial for
every fixed number $d$ of dimensions. We give several parameterized algorithms and
conclude bounds on possible kernelizations.
In Section \ref{def-sec-mkp}, we consider the multiple knapsack problem.
We give a dynamic programming solution and a pseudo-polynomial reduction
from {\sc 3-Partition} in order to show that the problem is not pseudo-polynomial in general, but for
every fixed number $m$ of knapsacks. Further we give parameterized algorithms and
bounds on possible kernelizations for several parameters.
In the final Section \ref{sec-con} we give some conclusions and an outlook
for further research directions.

\section{Preliminaries}\label{sec-pre}

In this section we recall basic notations for common
algorithm design techniques for hard problems
from the textbooks \cite{ACGKMP99},  \cite{DF13},  \cite{FG06}, and \cite{GJ79}.

%

\subsection{Parameterized Algorithms}

Within parameterized complexity we consider a two dimensional 
analysis of the computational complexity of a problem. Denoting the
input by $I$, the two considered dimensions are the input size $|I|$ 
and the value of a parameter $\kappa(I)$, see  \cite{DF13} and \cite{FG06} for surveys.

Let $\Pi$ be a decision problem and ${\mathcal I}$ the set of all instances
of $\Pi$. A {\em parameterization} or {\em parameter} of $\Pi$ is a 
mapping $\kappa: \mathcal{I} \to \IN$ that is polynomial time computable.
The value of the parameter $\kappa(I)$ is expected to be
small for all inputs $I\in  \mathcal{I}$.
A {\em parameterized problem} is a pair $(\Pi,\kappa)$, where $\Pi$
is a decision problem and $\kappa$ is a parameterization of $\Pi$.
For $(\Pi,\kappa)$ we will also use the abbreviation $\kappa$-$\Pi$.



\subsubsection{FPT-Algorithms}
An algorithm $A$ is an {\em fpt-algorithm with respect to $\kappa$}, if
there is a computable function $f: \IN \to \IN$ 
and a constant $c\in \IN $
such that for every instance $I\in {\mathcal I}$ the running time
of $A$ on $I$ is at most  
$$f(\kappa(I))\cdot |I|^c$$ or 
equivalently at most $f(\kappa(I))+ |I|^c$, see \cite{FG06}. 
For the case where $f$ is also a polynomial, $A$ is denoted as  
{\em polynomial fpt-algorithm with respect to $\kappa$}.

A parameterized problem $(\Pi,\kappa)$  
belongs to the class $\fpt$ and is called {\em fixed-parameter tractable},
if there is an fpt-algorithm with respect to $\kappa$ which decides $\Pi$.
Typical running times of an fpt-algorithm w.r.t.\ parameter $\kappa$
are $2^{\kappa(I)} \cdot |I|^2$ and $\kappa(I)! \cdot |I|^3 \cdot \log(|I|)$.

A parameterized 
problem $(\Pi,\kappa)$ belongs to the class $\pfpt$ and  is 
{\em polynomial fixed-parameter tractable} (cf.~\cite{CHKX07}),
if there is a polynomial fpt-algorithm with respect to $\kappa$ which 
decides $\Pi$. 

Please note that polynomial fixed-parameter tractability does not necessarily 
imply polynomial time computability for the decision problem in general.
A reason for this is that within  integer-valued problems 
there are parameter values $\kappa(I)$ which are
larger than any polynomial in the instance size $|I|$. An example is parameter
$\kappa(I)=c$ for problem {\sc Knapsack} in Section \ref{sec-para-kp}.
On the other hand, for small parameters polynomial fixed-parameter tractability
leads to polynomial time computability.

\begin{observation}
Let $(\Pi,\kappa)$ be some parameterized problem and $c$ be some constant such that for
every instance $I$ of $\Pi$ it holds $\kappa(I)\in \bigo(|I|^c)$. Then the existence 
of a polynomial fpt-algorithm with respect to $\kappa$ implies
a polynomial time algorithm for $\Pi$.
\end{observation}

\begin{proof}
Let $A$ be some polynomial fpt-algorithm with respect to $\kappa$ for $(\Pi,\kappa)$.
Then $A$ has a running time of $\bigo(\kappa(I)^d\cdot |I|^{d'})$ for two
constants $d$, $d'$. Since $\kappa(I)\in \bigo(|I|^c)$ we obtain a running time
which is polynomial in $|I|$.
\end{proof}

In order to state lower bounds we give the following corollary.

\begin{corollary}\label{cor-pfpt-p}
Let $(\Pi,\kappa)$ be some parameterized problem and $c$ be some constant 
such that $\Pi$ is NP-hard and for
every instance $I$ of $\Pi$ it holds $\kappa(I)\in \bigo(|I|^c)$. 
Then there is no polynomial fpt-algorithm with respect to $\kappa$ 
for $(\Pi,\kappa)$, unless  $\p=\np$.
\end{corollary}

\subsubsection{XP-Algorithms}
An algorithm $A$ is an {\em xp-algorithm with respect to $\kappa$}, if
there are two computable functions $f,g: \IN \to \IN$ 
such that for every instance $I\in {\mathcal I}$ the running time
of $A$ on $I$ 
is at most  
$$f(\kappa(I))\cdot |I|^{g(\kappa(I))}.$$

A parameterized problem  $(\Pi,\kappa)$ 
belongs to the class $\xp$ and is called {\em slicewise polynomial},
if there is an xp-algorithm with respect to $\kappa$ which decides $\Pi$.
Typical running times of an xp-algorithm w.r.t.\ parameter $\kappa$
are  $2^{\kappa(I)}\cdot|I|^{\kappa(I)^2}$ and $|I|^{\kappa(I)+2}$.

\subsubsection{Fixed-parameter intractability}
In order to show fixed-parameter intractability, it is useful to 
show the hardness with respect  to one of the classes $\w[t]$ for some 
$t\geq 1$, which were introduced by  Downey and Fellows \cite{DF99} 
in terms of weighted satisfiability problems on classes of circuits. 
The following relations -- the so called {\em $\w$-hierarchy} -- hold 
and all inclusions are assumed to be strict.
$$\pfpt\subseteq\fpt \subseteq \w[1] \subseteq \w[2] \subseteq \ldots \subseteq \xp$$

In the case of hardness\footnote{In this paper we
will consider several combined parameters for knapsack problems, e.g.
all profits, sizes, or capacities since the complexity of a parameterization by only
one of them remains open.} with respect to some parameter $\kappa$ a natural
question is whether the problem remains hard for {\em combined} parameters, i.e.\
parameters $(\kappa_1,\ldots,\kappa_r)$ that consists of $r\geq 2$ parts of the input. 
The given notations can be carried over to combined parameters, e.g.
an fpt-algorithm with respect to $(\kappa_1,\ldots,\kappa_r)$ is an algorithm of
running time $f(\kappa_1(I),\ldots,\kappa_r(I))\cdot |I|^c$ for  some constant $c$ and some computable function depending only on $\kappa_1,\ldots,\kappa_r$.


\subsubsection{Kernelization}\label{sec-kern} 

Next we consider the question, whether the 
sizes and the values can be reduced, such that their bit-length 
is bounded polynomially in a given parameter.

Let $(\Pi,\kappa)$ be a parameterized problem,  
${\mathcal I}$ the set of all instances
of $\Pi$ and $\kappa: \mathcal{I} \to \IN$ a para\-meter\-ization for $\Pi$.
A polynomial time transformation  $f:{\mathcal I} \times \IN \to {\mathcal I}
\times \IN$ is called a  {\em kernelization} for $(\Pi,\kappa)$, if $f$ maps
a pair $(I,\kappa(I))$ to a pair $(I',\kappa(I'))$, 
such that the following three properties
hold.

\begin{itemize}
\item For all $I\in {\mathcal I}$ it holds $I$ is a 
  yes-instance for $\Pi$ if and only if $I'$ is a 
  yes-instance for $\Pi$.

\item $\kappa(I')\leq \kappa(I)$.\footnote{In some recent works 
the restriction that the value of the new parameter is at most the value of the old parameter 
was relaxed, see \cite{CFKLMPPS15}.}

\item There is some function $f': \IN \to  \IN$, such that 
$|I'|\leq f'(\kappa(I))$.\footnote{The size of $I'$ depends only on the parameter $\kappa(I)$ and
not on the size of $I$.}
\end{itemize}
The pair $(I',\kappa(I'))$ is called {\em kernel} for $(\Pi,\kappa)$
and $f'(\kappa(I))$ is the {\em size of the kernel}. 
If $f'$ is a polynomial, linear, or constant function of $\kappa$, we say
$(I',\kappa(I'))$ is a {\em polynomial}, {\em linear}, or {\em constant 
kernel}, respectively, for $(\Pi,\kappa)$.

Next we show that fpt-algorithms lead to kernels. Although the existence is
well known (Theorem 1.39 in \cite{FG06}), we give a proof since we need the 
kernel size  later on.

\begin{theorem}\label{th-ker} 
Let $(\Pi,\kappa)$ be some parameterized problem.
If there is an fpt-algorithm that solves $(\Pi,\kappa)$ for every
instance $I$ in time $\bigo(f(\kappa(I))\cdot |I|^{c})$, then 
$\Pi$ is decidable and there is a 
kernel of size $\bigo(f(\kappa(I))$ for $(\Pi,\kappa)$.
\end{theorem}

\begin{proof}
Let $(\Pi,\kappa)\in\fpt$ and $A$ be an fpt-algorithm 
with respect to $\kappa$ which runs on input $I$ in time 
$f(\kappa(I))\cdot |I|^{c}$ for some function $f$ and some constant $c$. 
W.l.o.g.~we assume that there is one constant size yes-instance $I_0$ and one 
constant size no-instance $I_1$ of $\Pi$.

The following algorithm $A'$ computes a kernelization for $(\Pi,\kappa)$.
Algorithm $A'$ simulates $|I|^{c+1}$ steps of
algorithm $A$. If during this time $A$ stops and accepts or rejects, 
then $A'$ chooses $I_0$ or $I_1$, respectively, as the kernel.
Otherwise we know that $|I|^{c+1} \leq f(\kappa(I)) |I|^{c}$ and thus $|I|\leq f(\kappa(I))$ and
$A'$ states $(I,\kappa(I))$ as kernel.

Algorithm $A'$ has a running time in $\bigo(|I|^{c})$ and 
leads to a kernel of size $|I_0|+|I_1|+f(\kappa(I)) \in \bigo(1) \cup \bigo(1)\cup \bigo(f(\kappa(I)))\subseteq \bigo(f(\kappa(I)))$.
\end{proof}

The existence of an fpt-algorithm is even equivalent to the existence of a
kernelization for decidable\footnote{Bodlaender \cite{Bod09}
gives an example which shows that the condition that the problem is 
decidable is necessary.} problems \cite{DF13,FG06,Nie06}.

\begin{theorem}[Theorem 1.39 of \cite{FG06}]\label{th-fg} For every
parameterized problem $(\Pi,\kappa)$ the following properties
are equivalent:
\begin{enumerate}[(1)]
\item $(\Pi,\kappa)\in\fpt$
\item $\Pi$ is decidable and $(\Pi,\kappa)$ has a kernelization.
\end{enumerate}
\end{theorem}

Thus for fixed-parameter tractable problems the existence of
kernels of {\em polynomial} size are of special interest. 
For a long time polynomial kernels only were known for
parameterized problems obtained from optimization problems with
the standard parameterization (i.e.~problem $k$-$\Pi$ defined 
in Section \ref{sec-approx}). 
In this paper we will give a lot of examples for 
further parameters which lead to polynomial kernels
for knapsack problems.

For the special case where $f$ is a polynomial
Theorem \ref{th-ker}
implies that polynomial fpt-algorithms lead to polynomial kernels.

\begin{corollary}\label{th-polyker} 
Let $(\Pi,\kappa)$ be some parameterized problem.
If $(\Pi,\kappa)\in\pfpt$, then $\Pi$ is decidable and there is a 
poly\-no\-mial kernel for $(\Pi,\kappa)$.
\end{corollary}

A remarkable difference to the relation of Theorem \ref{th-fg} is
that the reverse direction of Corollary \ref{th-polyker} 
does not hold true, unless $\p=\np$. This can be shown by the knapsack problem
parameterized by the number of items $n$, $n$-{\sc KP} for short. By Theorem \ref{maintheorem2}
there is a polynomial kernel  of size $\bigo(n^4)$ for $n$-{\sc KP}
but by Corollary \ref{cor-pfpt-p} there is no polynomial fpt-algorithm for  $n$-{\sc KP}.

The used transformation $f$ for the proof of Corollary \ref{th-polyker} 
runs in time $\bigo(|I|^{d'+1})$, while a kernelization may have a 
transformation $f$ which runs in polynomial time in $|I|$ {\em and} $\kappa(I)$.
This will be exploited within the following characterization of
problems allowing kernels of constant size.

\begin{theorem}\label{th-me} For every
parameterized problem $(\Pi,\kappa)$ the following properties
are equivalent:
\begin{enumerate}[(1)]
\item $(\Pi,\kappa)\in\pfpt$
\item $\Pi$ is decidable and $(\Pi,\kappa)$ has a kernel of $\bigo(1)$ size.
\end{enumerate}
\end{theorem}

\begin{proof}
(1) $\Rightarrow$ (2): Let $(\Pi,\kappa)\in\pfpt$ and $A$ be a polynomial fpt-algorithm 
with respect to $\kappa$ which runs on input $I$ in time 
$\kappa(I)^d\cdot |I|^{d'}$ for two constants $d$ and $d'$. W.l.o.g.~we assume
that there is one constant size yes-instance $I_0$ and one 
constant size no-instance $I_1$ of $\Pi$. 
We run $A$ on input $I$ and instead of deciding we transform the input 
to the yes- or no-instance of bounded size. This leads to a kernel of constant size
and the algorithm uses polynomial time in $|I|$ and $\kappa(I)$.

(2) $\Rightarrow$ (1): If $(\Pi,\kappa)$ has a kernel of size $c\in \bigo(1)$, we can
solve the problem on input $I$ as follows. First we compute the kernel 
$(I',\kappa(I'))$ in polynomial time
w.r.t. $|I|$ and $\kappa(I)$. Then we check, whether the instance $I'$ belongs to the set of
yes-instances of size at most $c$, which does not dependent on the input.
\end{proof}

The special case that we take the parameter in unary, i.e. $|I|+\kappa(I)\in\Theta(|I|)$, 
was mentioned in \cite{Bod09}. Then for some parameterized problem $(\Pi,\kappa)$ 
it holds that $\Pi$ belongs to $\p$
if and only if it has a kernel of size $\bigo(1)$.

By Theorem \ref{th-me} and Corollary \ref{cor-pfpt-p} we 
obtain the following result.

\begin{corollary}\label{cor-pfpt-p2}
Let $(\Pi,\kappa)$ be some parameterized problem and $c$ be some constant 
such that $\Pi$ is NP-hard and for
every instance $I$ of $\Pi$ it holds $\kappa(I)\in \bigo(|I|^c)$. 
Then there is no kernel of $\bigo(1)$ size with respect to $\kappa$ 
for $(\Pi,\kappa)$, unless  $\p=\np$.
\end{corollary}

We have shown that every $(\Pi,\kappa)\in \pfpt$ has for every instance
$I$ a kernel of size $\kappa(I)^{\bigo(1)}$ 
using a kernelization of running time $|I|^{\bigo(1)}$. 
Further we have shown that every $(\Pi,\kappa)\in \pfpt$ has for every instance
$I$  a kernel of size $\bigo(1)$ 
using a kernelization of running time $(|I|+\kappa(I))^{\bigo(1)}$. 
We want to have a closer look at the differences between these two types of
kernelizations.

\begin{theorem}\label{th-me1} For every
parameterized problem $(\Pi,\kappa)$ the following properties
are equivalent:
\begin{enumerate}[(1)]
\item $(\Pi,\kappa)\in\fpt$
\item $\Pi$ is decidable and $(\Pi,\kappa)$ has a kernelization  of running time 
$|I|^{\bigo(1)}$.
\item $\Pi$ is decidable and $(\Pi,\kappa)$ has a kernelization  of running time 
$(|I|+\kappa(I))^{\bigo(1)}$.
\end{enumerate}
\end{theorem}

\begin{proof}
(1) $\Leftrightarrow$ (2): Proof of Theorem 1.39 in \cite{FG06}.
(1) $\Leftrightarrow$ (3): Proof of Theorem 1 in \cite{Bod09}.
\end{proof}

That is, the existence of a kernel found by a kernelization  of running time which is 
polynomial in $|I|$ is equivalent to the existence of a kernelization of running time 
which is polynomial in $|I|$ and $\kappa(I)$. It remains open whether this is also 
the case for the existence of polynomial kernels.

For the case of constant kernels the proof of Theorem \ref{th-me} uses
a kernelization  of running time which is polynomial in $|I|$ and $\kappa(I)$. 
This is really necessary, which can be seen as follows. The {\sc Knapsack} problem
parameterized by the capacity $c$, $c$-{\sc KP} for short, is in $\pfpt$
by Theorem \ref{maintheorem}. But if there would be a kernel for  $c$-{\sc KP} of size $\bigo(1)$
found by a kernelization  of running time  $|I|^{\bigo(1)}$ then the part (2) $\Rightarrow$ (1) of the proof
of Theorem \ref{th-me} implies a polynomial algorithm for {\sc Knapsack}.

A very useful theorem for finding 
kernels of knapsack problems with respect to parameter number of items $n$ is 
the following result of Frank and Tardos on
compressing large integer values to smaller ones.\footnote{We use the notations
appearing in the knapsack problems when stating the result of \cite{FT87}.}

\begin{theorem}[\cite{FT87}]\label{th-ft87}
Given a vector $(s_1,\dots, s_n, c)\in\Q^{n+1}$ and an integer $\ell\in\N_0$,
there exists an algorithm that computes a vector $(\tilde{s}_1,\dots,
\tilde{s}_n, \tilde{c})\in\Z^{n+1}$ in polynomial time, such that\footnote{Please note that in Theorem \ref{th-ft87} for some number $x$, the notation $|x|$ gives its absolute value.} 
$$\max\{
|\tilde{c}|, |\tilde{s}_j| : 1\leq
j\leq n\}\leq 2^{4(n+1)^3} (\ell+2)^{(n+1)(n+3)}$$ and
$$\sgn((s_1,\dots, s_n, c)\cdot (x_1,\dots, x_n,
x_{n+1}))=\sgn((\tilde{s}_1,\dots, \tilde{s}_n, \tilde{c})\cdot (x_1,\dots,
x_n, x_{n+1}))$$ 
for all $(x_1,\dots, x_n, x_{n+1})\in \Z^{n+1}$ with
$\sum_{j=1}^{n+1} |x_j| \leq \ell+1$.
\end{theorem}

By choosing vectors $(x_1,\dots, x_n, x_{n+1})=(1, 0, \dots),$ $\dots,$ $(0,
\dots, 0, 1)$ we immediately see that for each
$(s_1,\dots, s_n, c)\in\N^{n+1}$ there also is $(\tilde{s}_1,\dots,
\tilde{s}_n, \tilde{c})\in\N^{n+1}$.
This result can be used to equivalently replace equations and inequalities for $\odot \in\{=,\leq,\geq\}$
$$ s_1 x_1+s_2x_2+\dots+s_n x_n - c \odot 0$$
with
$$ \tilde{s}_1 x_1+\tilde{s}_2x_2+\dots+\tilde{s}_n x_n - \tilde{c} \odot
0$$
by choosing vector $(x_1,\dots,x_{n+1})=(x_1,\dots,x_n,-1)$ and $\ell$ such that
$\sum_{j=1}^{n} |x_j| \leq \ell$.

\subsection{Approximation Algorithms}\label{sec-approx}

Let $\Pi$ be some optimization problem and $I$ be some
instance of $\Pi$. By $OPT(I)$ we denote the value
of an optimal solution for $\Pi$ on input $I$. An {\em approximation
algorithm} $A$ for $\Pi$ is an algorithm which returns
a feasible solution for $\Pi$.
The value of the solution of $A$ on input $I$ is denoted
by $A(I)$.  
An approximation algorithm $A$ has {\em relative performance guarantee $\ell$}, 
if 
$$\max\bigg\{\frac{A(I)}{OPT(I)},\frac{OPT(I)}{A(I)} \bigg\}\leq \ell$$
holds for every instance 
$I$ of $\Pi$.

A {\em polynomial-time approximation scheme (PTAS)} for $\Pi$ is 
an algorithm $A$, for which the input consists of an instance of $\Pi$ and 
some $\epsilon$, $0<\epsilon <1$,
such that for every fixed $\epsilon$  algorithm  $A$ is a polynomial time 
approximation algorithm with relative performance  guarantee $1+\epsilon$.
An {\em efficient polynomial-time approximation scheme (EPTAS)} 
is a PTAS running in time $f(\nicefrac{1}{\epsilon})\cdot|I|^{c}$,
for some  computable function $f$ and some constant $c\in \IN $.
A {\em fully polynomial-time approximation scheme (FPTAS)} is  a PTAS 
running in time $(\nicefrac{1}{\epsilon})^{c}\cdot |I|^{c'}$, for two constants
$c$ and $c'$.
Obviously every FPTAS is an EPTAS and every EPTAS is a PTAS.

Next we recall relations between the existence of approximation schemes for 
optimization problems and fixed-parameterized algorithms.

Given some optimization problem $\Pi$ the corresponding decision problem of $\Pi$ is 
obtained by adding an integer $k$ to the input of $\Pi$ and changing the task 
into the question, whether the size of an optimal solution is at least 
(for maximization problems) or at most (for minimization problems) $k$. 
By choosing the threshold value $k$ as a parameter we obtain the so-called 
{\em standard parameterization} $k$-$\Pi$ of the so-defined decision problem.

\begin{desctight}
\item[Name] $k$-$\Pi$
\item[Instance] An instance $I$ of $\Pi$ and an integer $k$.
\item[Parameter] $k$
\item[Question] Is there a solution such that $OPT(I)\geq k$ (for a maximization problem
$\Pi$) or $OPT(I)\leq k$ (for a minimization problem $\Pi$)?
\end{desctight}

There are two useful connections between the existence of special 
PTAS for some optimization problem $\Pi$ and fpt-algorithms
for $k$-$\Pi$.

\begin{theorem}[\cite{CT97}, Proposition 2 in \cite{Mar08}]\label{theorem-fg}
If some optimization problem $\Pi$ has an EPTAS with
running time $\bigo(|I|^c \cdot f(\nicefrac{1}{\epsilon}))$, then 
there is an fpt-algorithm that solves 
the standard parameterization  $k$-$\Pi$ of the corresponding decision problem 
in time $\bigo(|I|^c \cdot f(2k))$.
\end{theorem}

\begin{theorem}[\cite{CC97}]\label{theorem-cc97}
If some optimization problem $\Pi$ has an  FPTAS with
running time $\bigo(|I|^c \cdot (\nicefrac{1}{\epsilon})^{c'})$, then 
there is a polynomial fpt-algorithm that solves 
the standard parameterization  $k$-$\Pi$ of the corresponding decision problem 
in time $\bigo(|I|^c \cdot (2k)^{c'})$.  
\end{theorem}

The main idea in the proofs of Theorem \ref{theorem-fg} and Theorem \ref{theorem-cc97}
is that the given approximation scheme for $\epsilon=\nicefrac{1}{2k}$ for optimization problem $\Pi$
leads to an fpt-algorithm that solves 
the standard parameterization of the corresponding decision problem $k$-$\Pi$
with the given running time. 
For so-called scalable optimization problems  (cf.~\cite{CHKX07}) the reverse
direction of Theorem \ref{theorem-cc97} also holds true.

By the definition, the existence of an approximation scheme for some 
optimization problem $\Pi$ applies the fundamental parameter 
$\kappa(I)=\nicefrac{1}{\epsilon}$  measuring the goodness of approximation.
Every PTAS provides for a fixed error  $\epsilon$ a polynomial time 
algorithm. Since these algorithms are not very practical, the question
arises whether $\nicefrac{1}{\epsilon}$ can be taken out of the exponent
of the input size. This is the case if the PTAS is even an EPTAS.
A formal method to combine the error bound $\epsilon$ and decision problems 
is the so-called gap version of an optimization problem,
which was introduced by Marx in \cite{Mar08}.

\subsection{Pseudo-polynomial Algorithms}

Let $\Pi$ be some optimization or decision problem and ${\mathcal I}$ 
the set of all instances of $\Pi$. For some $I\in {\mathcal I}$ we 
denote by $\max(I)$
the value of the largest number occurring in $I$.
An algorithm $A$ is  {\em pseudo-polynomial}, if
there is a polynomial $p: \IN \times \IN \to \IN$ 
such that for every instance $I$ the running time
of $A$ on $I$ is at most $p(|I|,\max(I))$, see \cite{GJ79}.
A problem $\Pi$ is {\em pseudo-polynomial} if it can be solved 
by a  pseudo-polynomial algorithm.

\begin{definition}[\cite{GJ79}]\label{def-str-np-h}
A problem is {\em strongly NP-hard}, if it remains NP-hard, if 
all of its numbers are bounded by a polynomial in the length 
of the input.
\end{definition}

\begin{theorem}[\cite{ACGKMP99}]\label{th-strongnp-pseudo}
If some problem $\Pi$  is strongly NP-hard, then $\Pi$  is not 
pseudo-polynomial.
\end{theorem}

The notation of pseudo-polynomial algorithms can be
carried over to parameterized algorithms \cite{Hro04a}. For an instance
$I$ of some problem $\Pi$ the function $\val(I)$ is defined by the
maximum length of the binary encoding of all numbers in $I$. 
Since the binary coding of an integer $w$ has length $1+ \lfloor\log_2(w)\rfloor$, 
the relation $\max(I)\leq 2^{\val(I)}$ holds and function $\val$
is a parameter for $\Pi$. Further, if there is a pseudo-polynomial
algorithm $A$ for $\Pi$ then there is a polynomial $p$, such that for every
instance $I$ of $\Pi$ the running time of $A$ can be bounded by
$$\bigo(p(|I|,\max(I)))\subseteq \bigo(p(|I|,2^{\val(I)})).$$
This implies that there are constants $c_1$ and $c_2$ such that 
the running time of $A$ can be bounded by $2^{c_1\cdot \val(I)}\cdot |I|^{c_2}$
and thus $A$ is an fpt-algorithm with respect to parameter $\val(I)$.

\begin{theorem}\label{th-val-para}
For every pseudo-polynomial problem $\Pi$
there is an fpt-algo\-ri\-thm that solves  $\val$-$\Pi$
in time $\bigo(2^{c_1\cdot \val(I)}\cdot |I|^{c_2})$.
\end{theorem}

\section{Knapsack Problem}\label{sec-kp}

The simplest of all knapsack problems is defined as follows.

\begin{desctight}
\item[Name] {\sc Max Knapsack} ({\sc Max KP})

\item[Instance] A set $A=\{a_1,\ldots,a_n\}$ of $n$ items, 
for every item $a_j$, there 
is a size of $s_j$ and a profit of $p_j$. 
Further there is a capacity $c$ 
for the knapsack.

\item[Task] Find a subset $A'\subseteq A$ such that 
the total profit of $A'$  is maximized 
and 
the total size of $A'$ is at most $c$.
\end{desctight}

The {\sc Max Knapsack}  problem can be approximated very good, since it allows
an FPTAS \cite{IK75} and thus can be
regarded as one of the easiest hard problems.

In this paper, the parameters $n$, $p_j$, $s_j$, and $c$ are assumed to be positive integers, i.e.
they belong to the set $\{1,2,3,\ldots\}$.
Let $s_{\max}=\max_{1\leq j \leq n}s_j$ and $p_{\max}=\max_{1\leq j \leq n}p_j$. 
The same notations are also used for $\min$ instead of $\max$.
In order to avoid trivial solutions we assume that $s_{\max}\leq c$  and that
$\sum_{j=1}^{n}s_j>c$.

For some instance $I$ its size $|I|$ can be bounded by 
the number of items and the binary encoding of all numbers in $I$ (cf.~\cite{GJ79}).
$$
\begin{array}{lcl}
|I|& =    & n + \sum_{j=1}^{n}(1+\lfloor\log_2(s_j)\rfloor) +  \sum_{j=1}^{n}(1+\lfloor\log_2(p_j)\rfloor)+ 1+\lfloor\log_2(c)\rfloor\\
   &  \in & \bigo(n+ \sum_{j=1}^{n}\log_2(s_j) +  \sum_{j=1}^{n}\log_2(p_j)+ \log_2(c)) \\
   & =    & \bigo(n+ n\cdot \log_2(s_{\max}) +  n\cdot \log_2(p_{\max})+ \log_2(c)) \\
\end{array}
$$
The size of the input is important for the analysis of running times.

In order to show fpt-algorithms and kernelizations we frequently will 
use bounds on the size of the input and on the size of a
solution of our problems. 
Let $I_1$ be a knapsack instance on
item set $A_1$. Instance $I_2$ on item set $A_2$ is a {\em reduced 
instance} for $I_1$ if 
\begin{inparaenum}[(1.)]
\item
$A_2\subseteq A_1$ and 
\item 
$OPT(I_1)=OPT(I_2)$.
\end{inparaenum}
For {\sc Max KP} we can bound the number and sizes of the items of an instance as
follows.

\begin{lemma}[Reduced Instance]\label{le-blb}
Every instance of {\sc Max KP} 
can be transformed into a reduced instance, such that  
$n\in \bigo(c\cdot \log(c))$.
\end{lemma}

\begin{proof}
In order to avoid trivial solutions we assume, that there is no item
in $A$, whose size is larger than the capacity $c$, i.e.\ 
$s_j\leq c$ for every $1\leq j \leq n$.
Further for $1 \leq s\leq c$
we can assume that there are at most $n_s:=\lfloor\frac{c}{s}\rfloor$  
items of size $s$ in $A$.

By the harmonic series we always can bound the number $n$ of items in $A$ by
$$n\leq \sum_{s=1}^{c} n_s= \sum_{s=1}^{c}  \bigg\lfloor\frac{c}{s}\bigg\rfloor \leq \sum_{s=1}^{c} \frac{c}{s} =  c\cdot \sum_{s=1}^{c} \frac{1}{s}  <  c\cdot (\ln(c)+1) \in \bigo(c\cdot \log (c)).$$

If we have given an instance $I_1$ for {\sc Max KP}  with more than 
the mentioned number $n_s$ of items of size $s$ for some  $1 \leq s\leq c$, we remove
all of them except the $n_s$ items of the highest profit. 
The new instance $I_2$ satisfies $n\in \bigo(c\cdot \log(c))$ and 
is a reduced instance of $I_1$. 
\end{proof}

The latter result is useful in Remark \ref{rem-c-kern}.
We have shown an alternative proof for Lemma \ref{le-blb} in \cite{GRY16}.

Since the capacity $c$ and the sizes of our items are positive integers,
every solution $A'$ of some instance of  {\sc Max KP}
even contains at most $c$ items. But this observation does not
lead to a reduced instance. In order so solve the problem using this bound 
one has to consider $\binom{n}{c}\in \bigo(n^c)$ many possible
subsets of $A$ which is much more inefficient than the dynamic programming approach
mentioned in Theorem \ref{t-dy2}.

\subsection{Binary Integer Programming}

Integer programming is a powerful tool, studied for over 50 years, that can
be used to define a lot of very important optimization problems \cite{JLN10}. 
{\sc Max KP} can be formulated using a boolean variable $x_j$ for every item 
$a_j\in A$, indicating whether or not $a_j$ is chosen into the solution $A'$,
by a so-called binary integer program (BIP).
\begin{eqnarray}
&\text{max } & \sum_{j=1}^{n} p_{j}  x_{j}\label{lp-kp1}   \\
&\text{s.t. }& \sum_{j=1}^{n} s_j x_j  \leq c \label{lp-kp2} \\
&\text{and  }& x_j\in\{0,1\}  \text{ for  } j\in [n]\label{lp-kp3} 
\end{eqnarray}

For some positive integer $n$, let
$[n]=\{1,\ldots,n\}$ be 
the set of all positive integers between $1$ and $n$.
We apply BIP versions for our knapsack problems to obtain 
parameterized algorithms (Theorem \ref{maintheorem}).

\subsection{Dynamic Programming Algorithms}

Dynamic programming  solutions for {\sc Max KP} are well known. 
The following two results can be found in
the textbook \cite{KPP10}.

\begin{theorem}[Lemma 2.3.1 of  \cite{KPP10}]\label{t-dy2}
{\sc Max KP} can be solved in time $\bigo(n\cdot c)$.
\end{theorem}

\begin{theorem}[Lemma 2.3.2 of  \cite{KPP10}]\label{t-dy1}
{\sc Max KP}  can be solved in time $\bigo(n\cdot U)\subseteq
\bigo(n \cdot \sum_{j=1}^{n}p_{j})\subseteq\bigo(n^2 \cdot p_{\max})$,
where $U$ is an upper bound on the value of an optimal solution.
\end{theorem}

Since for unary numbers the value of the number is equal to the
length of the number the running times of the two cited
dynamic programming solutions is even polynomial.
Thus {\sc Max KP} can be solved in 
polynomial time if all numbers are given in unary.
In this paper we assume that all numbers are encoded 
in binary.

\subsection{Pseudo-polynomial Algorithms}

Although {\sc Max KP} is a well known example for a pseudo-polynomial problem
we want to give this result for the sake of completeness.

\begin{theorem}\label{th-si-bug-pseudo}
{\sc Max KP} is pseudo-polynomial.
\end{theorem}

\begin{proof}
We consider the running time of the algorithm which proves Theorem \ref{t-dy1}.
In the running time 
$\bigo(n \cdot \sum_{j=1}^{n}p_{j})\subseteq\bigo(n^2 \cdot p_{\max})$ part
$n^2$ is polynomial in the input size and part $p_{\max}$ is polynomial
in the value of the largest occurring number in every input. (Alternatively
we can apply the running time of the algorithm cited in  Theorem \ref{t-dy2}.)
\end{proof}

\subsection{Parameterized Algorithms}\label{sec-para-kp}

Since {\sc Max KP} is an integer-valued problem 
defined on inputs of various informations, it makes
sense to consider parameterized versions of the problem. By 
adding a threshold value $k$ for  the profit to the instance
and choosing a parameter $\kappa(I)$ from this instance
$I$, we define the following parameterized problem.

\begin{desctight}
\item[Name] $\kappa$-{\sc Knapsack} ($\kappa$-{\sc KP})

\item[Instance]  A set $A=\{a_1,\ldots,a_n\}$ of $n$ items, 
for every item $a_j$, there 
is a size of $s_j$ and a profit of $p_j$. 
Further there is a capacity $c$ 
for the knapsack
and a positive integer $k$.

\item[Parameter] $\kappa(I)$

\item[Question] Is there a subset $A'\subseteq A$ such that 
the total profit of $A'$  is  at least $k$
and 
the total size of $A'$ is at most $c$.
\end{desctight}

For some instance $I$ of {\sc $\kappa$-KP} its size $|I|$ can be bounded by
$$ |I| \in \bigo(n+ n\cdot \log_2(s_{\max}) +  n\cdot \log_2(p_{\max})+ \log_2(c) + \log_2(k)) .$$

Next we give  parameterized algorithms for the knapsack problem.
The parameter $\sizevar(I)=|\{s_1,\ldots,s_n\}|$ counts the number of distinct item sizes
within knapsack instance $I$.

\begin{table}[ht]
\begin{center}
\begin{tabular}{|l||l|lcl|l|}
\hline
Parameter      &  class   & time  &&\\
\hline\hline
$k$            & $\in\pfpt$  & $\bigo(n^2 \cdot k)$ & $\subseteq$ &   $\bigo(|I|^2\cdot k)$  \\
\hline
$c$            & $\in\pfpt$  & $\bigo(n\cdot c)$ &   $\subseteq$ &   $\bigo(|I|\cdot c)$    \\
\hline
$p_{\max}$     & $\in\pfpt$  & $\bigo(n^2\cdot p_{\max})$ &  $\subseteq$ &   $\bigo(|I|^2\cdot p_{\max})$    \\
\hline
$s_{\max}$     & $\in\pfpt$  &  $\bigo(n^2 \cdot s_{\max})$ & $\subseteq$ &   $\bigo(|I|^2 \cdot s_{\max})$           \\
\hline
$n$            & $\not\in\pfpt$,  $\in\fpt$   & $\bigo(n\cdot 2^n)$ &  $\subseteq$ &   $\bigo(|I|\cdot 2^n)$ \\
\hline
$\val$         & $\not\in\pfpt$, $\in\fpt$    & $\bigo(n\cdot 2^{\val(I)})$ &  $\subseteq$ &   $\bigo(|I|\cdot 2^{\val(I)})$ \\
\hline
$\sizevar$ $s$    & $\not\in\pfpt$, $\in\fpt$    & &&  $\bigo(|I|^{\bigo(1)}\cdot s^{2.5s\cdot o(s)})$   \\
\hline
\end{tabular}
\end{center}
\caption{Overview of parameterized algorithms for KP}\label{tab-kp-survey}
\end{table}

\begin{theorem}\label{maintheorem}
There exist parameterized algorithms for the knapsack problem 
such that the running times of Table \ref{tab-kp-survey} 
hold true and the problems belong to the specified parameterized complexity classes.
\end{theorem}

\begin{proof}
For the standard parameterization $\kappa(I)=k$ we use the fact that
the {\sc Max KP} problem allows an FPTAS of running time 
$\bigo(n^2 \cdot \frac{1}{\epsilon})$, see  \cite{KPP10} for a survey. 
By Theorem \ref{theorem-cc97} we can use
this FPTAS for $\epsilon=\nicefrac{1}{2k}$ in order to obtain a polynomial 
fpt-algorithm that solves the standard parameterization of the corresponding 
decision problem in time $\bigo(n^2 \cdot 2k)=\bigo(n^2\cdot k)$.

For parameter $\kappa(I)=c$ we consider the dynamic programming algorithm shown in the proof 
of Theorem \ref{t-dy2} which has running time $\bigo(n\cdot c)$. 
From a parameterized point of view a polynomial fpt-algorithm that solves
$c$-{\sc KP}  follows.

For $\kappa(I)=p_{\max}$ we consider the algorithm shown in the proof of Theorem \ref{t-dy1}.
Since its running time is  
$\bigo(n \cdot \sum_{i=1}^{n}p_{i})\subseteq\bigo(n^2 \cdot p_{\max})$, from a 
parameterized point of view a polynomial fpt-algorithm follows that solves
$p_{\max}$-{\sc KP}.

For parameter $\kappa(I)=s_{\max}$ we distinguish the following two cases.
If for some $s_{\max}$-{\sc KP} instance it holds 
$\sum_{i=1}^{n}s_i \leq c$, then all items fit into the knapsack and 
we can choose $A'=A$ and verify whether $\sum_{i=1}^{n}p_i \geq k$. 
Otherwise we know that $c< \sum_{i=1}^{n}s_i\leq n\cdot s_{\max}$ and by the
algorithm shown in the proof of Theorem \ref{t-dy2} we
can solve the problem in time $\bigo(n \cdot c)\subseteq \bigo(n^2 \cdot s_{\max})$.

For parameter $\kappa(I)=n$ we can use
a brute force solution by checking for all $2^n$ possible subsets of $A$
the constraint (\ref{lp-kp2}) of the given BIP, which
leads to an algorithm of time complexity $\bigo(n\cdot 2^n)$.
Alternatively one can use  the result of \cite{Len83}, or its improved
version in \cite{Kan87}, which implies
that integer linear programming is fixed-parameter tractable for the 
parameter ''number of variables''. 
%
%
Thus by BIP (\ref{lp-kp1})-(\ref{lp-kp3}) the $n$-{\sc KP} problem is fixed-parameter 
tractable in time $\bigo(|I|\cdot n^{\bigo(n)})$.

For parameter $\kappa(I)=\val(I)$, i.e.\
the maximum length of the binary encoding of all numbers within
the instance $I$ we know by Theorem \ref{t-dy2} that problem $\kappa(I)$-{\sc KP}  
can be solved in $\bigo(n\cdot c) \subseteq \bigo(n \cdot 2^{\val(I)})$ and thus
is fixed-parameter trac\-table with respect to parameter $\kappa(I)=\val(I)$.

For parameter $\kappa(I)=\sizevar(I)$ we know from \cite{EKMR15} that
there is an fpt-algorithm that solves $\sizevar(I)$-{\sc KP}
in time $\bigo(s^{2,5s\cdot o(s)}\cdot |I|^{\bigo(1)})$, where $s=\sizevar(I)$.

\medskip
For parameters $\kappa(I)\in\{n,\val(I),\sizevar(I)\}$ for every
instance $I$ it holds $\kappa(I)\leq |I|$ and by Corollary \ref{cor-pfpt-p}
there is no  polynomial fpt-algorithm for $\kappa(I)$-{\sc KP}.
\end{proof}

There is even no algorithm of running time $2^{o(n)}$ for $n$-{\sc KP}, assuming the Exponential Time
Hypothesis.\footnote{The Exponential Time Hypothesis \cite{IPZ01} states that there does not exist 
an algorithm of running time  $2^{o(n)}$ for 3-SAT, where $n$ denotes the number of variables.
} Since such an algorithm would also imply an  algorithm of running time $2^{o(n)}$
for $n$-{\sc Subset Sum}, which was disproven in \cite{EKMR15}.

\subsection{Kernelizations}

Next we give kernelization bounds for the knapsack problem.

\begin{table}[ht]
\begin{center}
\begin{tabular}{|l||l|l|l|}
\hline
Parameter      & lower bound  & upper bound  \\
\hline\hline
$k$, $c$, $p_{\max}$, $s_{\max}$              & $\Theta(1)$  & $\Theta(1)$  \\
\hline
$n$            & $\omega(1)$  & $\bigo(n^4)$  \\
\hline
$\val$         & $\omega(1)$  & $\bigo(2^{\val(I)})$   \\
\hline
$\sizevar$     & $\omega(1)$  & $\bigo(s^{2,5s\cdot o(s)})$\\
\hline
\end{tabular}
\end{center}
\caption{Overview for kernel sizes of parameterized KP}\label{tab-kp-ker-survey}
\end{table}

\begin{theorem}\label{maintheorem2}
There exist kernelizations for the parameterized knapsack problem such that the upper bounds for the sizes
of a possible kernel in Table \ref{tab-kp-ker-survey} 
hold true.
\end{theorem}

\begin{proof}
For parameters $\kappa(I)\in\{k,c,p_{\max},s_{\max}\}$ we obtain by 
Theorem \ref{maintheorem} and Theorem \ref{th-me}  a kernel of constant
size for $\kappa(I)$-{\sc KP}.

For parameter  $\kappa(I)=n$ let $I$ be an instance of $n$-{\sc KP}.
We apply Theorem \ref{th-ft87} in order to 
equivalently replace the inequality
\begin{equation}
s_1 x_1+s_2x_2+\dots+s_n x_n - c \leq 0 \text{~~ by ~~}
 \tilde{s}_1 x_1+\tilde{s}_2x_2+\dots+\tilde{s}_n x_n - \tilde{c} \leq
0 \label{eq-ker-n-1}
\end{equation}
such that $\tilde{s}_1,\ldots,\tilde{s}_n,\tilde{c}$ are positive integers and
$$\max\{|\tilde{c}|, |\tilde{s}_j| : 1\leq
j\leq n\}\leq 2^{4(n+1)^3} (\ell+2)^{(n+1)(n+3)}.$$
In the same way we can replace the inequality
\begin{equation}
p_1x_2 +p_2 x_2 + \ldots + p_nx_n  - k \geq 0\text{~~ by ~~}
\tilde{p}_1x_2 +\tilde{p}_2 x_2 + \ldots + \tilde{p}_nx_n - \tilde{k} \geq 0 \label{eq-ker-n-2}
\end{equation}
such that $\tilde{p}_1,\ldots,\tilde{p}_n,\tilde{k}$ are positive integers and
$$\max\{|\tilde{k}|, |\tilde{p}_j| : 1\leq
j\leq n\}\leq 2^{4(n+1)^3} (\ell+2)^{(n+1)(n+3)}.$$
For the obtained instance $I'$ we can
bound $|I'|$ by the number of items $n$ and $2n+2$ numbers
of value at most $2^{4(n+1)^3} (\ell+2)^{(n+1)(n+3)}$.
Since we can assume $\ell\leq n$ for {\sc
$n$-KP}, this problem has a polynomial kernel of size 
$$\bigo\left(n+ (2n+2) \log_2\left( 2^{4(n+1)^3} (\ell+2)^{(n+1)(n+3)} \right) \right)
\subseteq \bigo(n^4 + \log_2(\ell+2)n^3) \subseteq \bigo(n^4).
$$

For parameter $\kappa(I)\in\{\val(I),\sizevar(I)\}$ the upper bounds follow by
Theorem \ref{maintheorem} and Theorem \ref{th-ker}.

\medskip
The lower bounds for the kernel sizes 
for parameters $\kappa(I)\in\{n,\val(I),\sizevar(I)\}$ hold, since for every
instance $I$ it holds $\kappa(I)\leq |I|$ and by Corollary \ref{cor-pfpt-p2}
there is no  kernel of constant size for $\kappa(I)$-{\sc KP}.
\end{proof}

Next we give some further ideas how to show kernels for  $c$-{\sc KP}.
Although the sizes are  non-constant, the ideas might be interesting on its own.

\begin{remark}\label{rem-c-kern}
\begin{enumerate}
\item Let $I$ be some instance of $c$-{\sc KP}. As in the proof of Theorem
\ref{maintheorem2} we apply Theorem \ref{th-ft87} in order to obtain
a polynomial kernel $I'$ of size $\bigo(n^4)$. By the proof of Lemma \ref{le-blb} 
we can transform $I'$ into a reduced
instance, such that $n \in \bigo(c \cdot \log_2(c))$ thus we obtain a kernel
of size $\bigo(c^4 \cdot \log^4_2(c))\subseteq \bigo(c^5)$.
Thus there is a kernel of size $\bigo(c^5)$ for $c$-{\sc KP}. 

\item Next we restrict to the case where $p_{\min}=1$.
Let $I$ be some instance of $c$-{\sc KP}. Its size can be bounded by
$|I|\in\bigo(n+ n\cdot \log_2(s_{\max}) +  n\cdot \log_2(p_{\max})+ \log_2(c))$.
By the proof of Lemma \ref{le-blb} we can transform $I$ into a reduced
instance, such that $n \in \bigo(c \cdot \log_2(c))$ and $s_{\max}\leq c$.
It remains to show that we can bound $p_{\max}$ by some function $f(c)$.

Therefor we observe that if $p_{\max}$ is greater than the sum
of the profits of the other items, which implies that there is only one item $a_j$ of
size $p_{\max}$, then $a_j$ must be included in any optimal solution $A'$. This allows
us to proceed with item set $A-\{a_j\}$ and capacity $c-s_{j}$.\footnote{Instances
with so called superincreasing profits $p_j> \sum_{j'=1}^{j-1} p_{j'}$ can be solved 
in polynomial time.} 

Thus if we sort the items  in ascending order w.r.t. the 
profits $1=p_1 \leq p_2 \leq \ldots \leq p_{n}$
we know that $p_i\leq \sum_{j=1}^{i-1}p_j$ for $2\leq i \leq n$. 
This implies $p_i\leq 2^{i-2} \cdot p_1 = 2^{i-2}$ for $2\leq i \leq n$ and for
$i=n$ we obtain $p_{\max} =  p_{n}\leq 2^{n-2}\in \bigo(2^{c\cdot \log_2(c)-2})$.
Thus we obtain
$$\begin{array}{lcl}
|I| &\in & \bigo(c\cdot \log_2(c) + c\cdot \log_2(c) \cdot  \log_2(c) + c\cdot \log_2(c) \cdot \log_2(2^{c\cdot \log_2(c)-2}) + \log_2(c)) \\
 &\subseteq & \bigo(c\cdot \log_2(c) + c\cdot \log_2(c)  \cdot  \log_2(c) + c\cdot \log_2(c) \cdot (c\cdot \log_2(c) -2) + \log_2(c)) \\
&\subseteq & \bigo(c^3)
\end{array}
$$
Thus there is a kernel of size $\bigo(c^3)$ for $c$-{\sc KP} if $p_{\min}=1$.

It remains open whether there is 
a (feasible) kernelization  which leads to $p_{\min}=1$.
We tried to divide all profits by $p_{\min}$ and round
the obtained values up or down or subtract $p_{\min}-1$
from all profits.
\end{enumerate}
\end{remark}

The existence of a kernel for $c$-{\sc KP} was also stated 
in Theorem 4.11 of \cite{Fer05} without 
giving a bound on the size of the kernel.
A randomized Turing kernel for knapsack w.r.t. parameter $n$ was shown  in \cite{NLZ12}. 

\section{Multidimensional Knapsack Problem}\label{sec-mkp}

Next we consider the knapsack problem for multiple dimensions.

\begin{desctight}
\item[Name] {\sc Max d-dimensional Knapsack} ({\sc Max d-KP})

\item[Instance] A set $A=\{a_1,\ldots,a_n\}$ of $n$ items and a number $d$ of dimensions.
Every item $a_j$ has a profit $p_j$ and for dimension $i$ the size $s_{i,j}$. 
Further for every dimension $i$ there is a capacity $c_i$.

\item[Task]  Find a subset $A'\subseteq A$ such that 
the total profit of $A'$  is maximized 
and 
for every dimension $i$ the total size of $A'$ is at most the capacity $c_i$.
\end{desctight}

In the case of $d=1$ the {\sc Max d-KP} problem corresponds to the {\sc Max KP} problem considered in
Section  \ref{sec-kp}.
For  {\sc Max d-KP}  there is no PTAS in general, since
by the proof of Theorem \ref{th-multi-pseudo} there is a 
PTAS reduction (cf.~Definition 8.4.1 in \cite{Weg05}) from
{\sc Max Independent Set}, which does not allow a PTAS, see  \cite{ACGKMP99}.
For every fixed dimension $d$ the problem  {\sc Max d-KP}
has a PTAS with running time  $\bigo(n^{\lceil\frac{d}{\epsilon} \rceil-d})$ by 
\cite{CKPP00}. In \cite{KS10} it has been shown, that 
there is no EPTAS for {\sc Max d-KP}, even for fixed 
$d=2$, unless $\fpt=\w[1]$. 
A recent survey for d-dimensional knapsack problem was given in 
\cite{Fre04}. A survey on  different types of non-para\-me\-terized algorithms for 
the d-dimensional knapsack problem can be found in \cite{Var12}.

Parameters $n$, $d$, $p_j$, $s_{i,j}$, and $c_i$ are assumed to be positive integers. 
As usually (cf.~\cite{KPP10}) we allow that  $s_{i,j}=0$ for some $1\leq i \leq d$, $1\leq j\leq n$  
if $\sum_{i=1}^{d}s_{i,j}\geq 1$ for every item $a_j$.\footnote{The possibility of $s_{i,j}=0$ is also needed
in the proof of  Theorem \ref{th-multi-pseudo}.}
Let $p_{\max}=\max_{1\leq j \leq n}p_j$, 
$s_{\max}=\max_{1\leq i \leq d, 1\leq j \leq n}s_{i,j}$, and
$c_{\max}=\max_{1\leq i \leq d}c_i$.
The same notations are also used for $\min$ instead of $\max$.
In order to avoid trivial solutions (cf. \cite{KPP10})
we assume that $s_{i,j}\leq c_i$ for all $j\in[n]$, $i\in [d]$ and
$\sum_{j=1}^n s_{i,j}\geq c_i$ for all $i\in [d]$.

For some instance $I$ of {\sc Max d-KP} its size $|I|$ can be bounded  as follows (cf.~\cite{GJ79}).
$$
\begin{array}{lcl}
|I|& = & n+ \sum_{j=1}^{n}(1+\lfloor\log_2(p_j)\rfloor) +  \sum_{i=1}^{d}\sum_{j=1}^{n}(1+\lfloor\log_2(s_{i,j})\rfloor)+ \sum_{i=1}^{d}(1+\lfloor\log_2(c_i)\rfloor)\\
   &  \in & \bigo(n\cdot d+ \sum_{i=1}^{n}\log_2(p_j) +  \sum_{i=1}^{d}\sum_{j=1}^{n}\log_2(s_{i,j})+ \sum_{i=1}^{d}\log_2(c_i)) \\
   &  = & \bigo(n \cdot d + n\cdot \log_2(p_{\max}) +  n\cdot d \cdot \log_2(s_{\max})+ d\cdot \log_2(c_{\max})) \\
\end{array}
$$



Next we give a bound on the number of items
in a reduced instance w.r.t.\ the capacities. The main idea is to 
identify the items with d-dimensional vectors of sizes, whose number
can be bounded.

\begin{lemma}[Reduced Instance]\label{le-red-port-allo1}
Every instance of {\sc Max d-KP} can be
trans\-formed into a reduced instance, such that  
$n\leq c_{\min} \cdot (\Pi_{i=1}^d (c_i+1)-1)\in \bigo((c_{\max}+1)^{d+1})$.
\end{lemma}

\begin{proof}
Let $I$ be an instance of {\sc Max d-KP} on item set $A$. For every item $a_j\in A$ we 
denote by $(s_{1,j},s_{2,j},\ldots,s_{d,j})$ its sizes in all $d$ dimensions.
By our capacities we can assume that every such $d$-tuple 
$(s_{1,j},s_{2,j},\ldots,s_{d,j})$ is contained in set  $S=\{(s_1,\ldots,s_d)\in \IN_0^d ~|~ 0\leq s_{i}\leq c_i, \sum_{i=1}^d s_i\geq 1\}$. Set $S$ contains
at most $\Pi_{i=1}^d (c_i+1)-1$ elements.

For every size vector $s\in S$ there can be at most $n_s:=\min_{i\in[d],s_i\neq 0}\lfloor \frac{c_i}{s_i}\rfloor$ items in $A$ of vector $s$ since when choosing an item $a_j$ into a solution $A'$
it contributes to every dimension and there is always one dimension $i$ such that the size
of $a_j$ in $i$ is positive.

Thus we always can bound the number $n$ of items in $A$ by
$$\begin{array}{lcl}
n &\leq &\sum_{s\in S} n_s \\
 & = & \sum_{s\in S} \min_{i\in[d],s_i\neq 0}\lfloor \frac{c_i}{s_i}\rfloor\\
 & \leq & \sum_{s\in S} \min_{i\in[d],s_i\neq 0} \frac{c_i}{s_i}\\
 & \leq & \sum_{s\in S} \min_{i\in[d]}c_i\\
  & \leq &  \min_{i\in[d]} c_i \cdot |S|\\
   & \leq &   c_{\min} \cdot (\Pi_{i=1}^d (c_i+1)-1)\\
\end{array}
$$

If we have given an instance $I_1$ for {\sc Max d-KP}  with more than 
the mentioned number $n_s$ of items of size vector $s$ for some $s\in S$, we remove
all of them except the $n_s$ items of the highest profit. 
The new instance $I_2$ satisfies $n\leq  c_{\min} \cdot (\Pi_{i=1}^d (c_i+1)-1)$ and 
is a reduced instance of $I_1$. 
\end{proof}

Since the capacities $c_i$, $i\in[d]$, and the sizes of our items are non-negative integers,
every solution $A'$ of some instance of  {\sc Max d-KP}
even contains at most $\sum_{i=1}^d c_i$ items. For the special case that all
sizes are positive every solution $A'$ of some instance of  {\sc Max d-KP}
even contains at most $c_{\min}$ items.
But these observations do not
lead  to reduced instances.

The {\em size} of a set is the number of its elements and the size of 
a number of sets is the size of its union.

\begin{lemma}[Bounding the size of a solution]\label{le-bd-sol-size-dkp} 
For every instance of {\sc Max d-KP}
there is a feasible solution $A'$ of profit at least $k$ 
if and only if
there is a feasible solution $A''$ of profit at least $k$, 
which has size at most $k$.
\end{lemma}

\begin{proof}
Let $I$ be an instance of $k$-{\sc d-KP}
and $A'$ be a solution which complies the capacities of every dimension
and the profit of $A'$ is at least $k$. 
Whenever there are at least $k+1$ items in $A'$ we can remove one
of the items of smallest profit $p'$ and obtain a solution $A''$
which still complies the capacities of every dimension. Further since
all profits are  positive integers, the
profit of $A''$ is at least $p'\cdot(k+1)-p'=p'\cdot k\geq k$.
\end{proof}

\subsection{Binary Integer Programming}

Using a boolean variable $x_j$ for every item $a_j\in A$,
indicating whether or not the item $a_j$ will be chosen into $A'$,
a binary integer programming (BIP) version of {\sc Max d-KP} is 
as follows.
\begin{eqnarray}
&\text{max } & \sum_{j=1}^{n} p_{j}   x_j                                   \label{lp-dkp1} \\
&\text{s.t. }& \sum_{j=1}^{n} s_{i,j} x_j  \leq c_i  \text{ for } i\in [d]  \label{lp-dkp2} \\
&\text{and  }& x_j\in\{0,1\}                         \text{ for } j\in [n]  \label{lp-dkp3} 
\end{eqnarray}

\subsection{Dynamic Programming Algorithms}

Dynamic programming  solutions for {\sc Max d-KP}
can be found in \cite{WN67} and \cite{KPP10}. 
The following result holds by the textbook \cite{KPP10}.

\begin{theorem}[Section 9.3.2 of  \cite{KPP10}]\label{t-dy3}
{\sc Max d-KP} can be solved in time 
$\bigo(n\cdot d \cdot \Pi_{i=1}^d c_i)\subseteq \bigo(n\cdot d \cdot (c_{\max})^d)$.
\end{theorem}

\subsection{Pseudo-polynomial Algorithms} \label{multip-pseudo}

The existence of pseudo-polynomial algorithms
for {\sc Max d-KP} depends on the assumption whether the number of 
dimensions $d$ is given in the input or is assumed to be fixed.

\begin{theorem}\label{th-multi-pseudo}
{\sc Max d-KP} is not pseudo-poly\-no\-mial.
\end{theorem}

\begin{proof}
Every NP-hard problem for which every instance $I$ only 
contains numbers $x$, such that the value of $x$ is polynomial bounded
in $|I|$ is strongly NP-hard (cf.~Definition \ref{def-str-np-h}) and
thus not pseudo-polynomial (cf.~Theorem \ref{th-strongnp-pseudo}).
To show that {\sc Max d-KP} is not pseudo-polynomial in general
we can use a pseudo-polynomial reduction (cf.~page 101 in \cite{GJ79}) 
from {\sc Max Independent Set}. The problem
is that of finding a maximum independent set in a graph $G=(V,E)$, 
i.e.\ a subset $V'\subseteq V$ 
such that no two vertices of $V'$ are adjacent and $V'$ has maximum size.

Let graph $G=(V,E)$ be an input for the {\sc Max Independent Set}
problem. For every vertex $v_j$ of $G$ we define an item $a_j$ 
and for every edge of $G$ we define a dimension, i.e.\ $d=|E|$,
within an instance $I_G$ for {\sc Max d-KP}. 
The profit of every item is equal to 1 and the capacity for every 
dimension is equal to 1, too. The size of an item within a dimension
is equal to 1, if the vertex corresponding to the item
is involved in the edge corresponding to the dimension, otherwise
the size is 0, see Figure \ref{red} and Example \ref{ex-proof}.

By this construction a maximum independent set $V'\subseteq V$ in $G$
corresponds to  a subset $A'$ of maximum profit within $I_G$. Further
since all $c_i=1$ for every edge we can choose at most one vertex
into an independent set, and thus a subset $A'$ of maximum profit within $I_G$
corresponds to a maximum independent set $V'\subseteq V$ in $G$.
Thus the size of a maximum 
independent set in $G$ is equal to the value of a maximum possible
profit within $I_G$. Since the value of the largest number in $I_G$
is polynomial bounded in the value of the largest number and 
the input size of our original instance $G$ we have found a pseudo-polynomial reduction.
\end{proof}

\begin{example}\label{ex-proof} Figure \ref{red} shows an example for the
reduction given in the proof of Theorem \ref{th-multi-pseudo}.
The graph on six vertices and seven edges defines the {\sc Max d-KP}
instance on six items and seven dimensions
shown in the table.
The maximum independent set $V'=\{v_1,v_4,v_6\}$ of size $3$ corresponds to 
the subset $A'=\{a_1,a_4,a_6\}$ of maximum profit $3$.
\end{example}

\begin{figure}[ht]
\begin{minipage}{45mm}
\centerline{\epsfig{figure=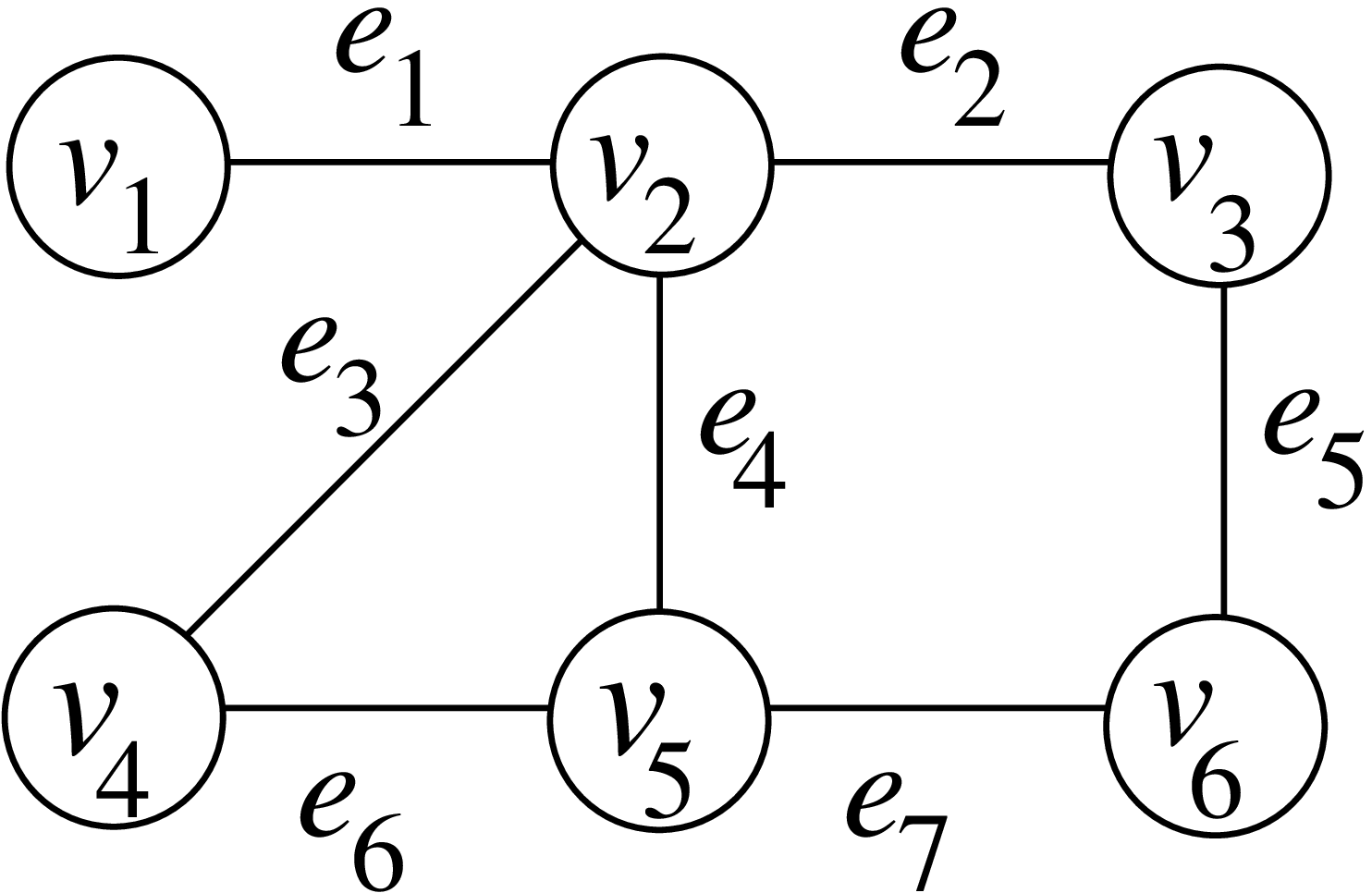,width=3.5cm}}
\end{minipage}
\hspace{1cm}
\begin{minipage}[b]{50mm}
$$
\begin{array}{l|c|c|c|c|c|c||c}
s_{i,j}  & j=1 & j=2 & j=3 & j=4 & j=5 & j=6 & c_i\\
\hline
i=1   &  1  &  1 &  0  & 0       & 0 & 0 & 1   \\
i=2   &  0  &  1 &  1  & 0       & 0 & 0 & 1 \\
i=3   &  0  &  1 &  0  & 1       & 0 & 0 & 1 \\
i=4   &  0  &  1 &  0  & 0       & 1 & 0 & 1\\
i=5   &  0  &  0 &  1  & 0       & 0 & 1 & 1 \\
i=6   &  0  &  0 &  0  & 1       & 1 & 0 & 1 \\
i=7   &  0  &  0 &  0  & 0       & 1 & 1 & 1 \\
\hline
\hline
p_j     & 1 & 1  & 1  & 1   & 1 & 1      \\
\end{array}
$$
\end{minipage}
\caption{Reduction in the proof of Theorem \ref{th-multi-pseudo}}
\label{red}
\end{figure}

\begin{remark} The proof of Theorem \ref{th-multi-pseudo} shows that 
{\sc Max d-KP} is not pseudo-poly\-no\-mial for $n\leq d$.
In order to show the same result for the more common case $d<n$
we can consider special  {\sc Max Independent Set} instances by replacing
every graph $G=(V,E)$ by $G'=(V',E)$, where $V'=V\cup\{v_e~|~e\in E\}$.
\end{remark}

%
%

\begin{theorem}\label{th-multi-pseudo2}
For every fixed $d$ there is a pseudo-polynomial algorithm that solves 
{\sc Max d-KP} in time $\bigo(n\cdot d \cdot (c_{\max})^d)$.
\end{theorem}

\begin{proof}
We consider the algorithm with running 
time $\bigo(n\cdot d \cdot (c_{\max})^d)$ of
Theorem \ref{t-dy3}. If $d$ is assumed to be fixed, $n\cdot d$
is polynomial in the input size and $(c_{\max})^d$ is polynomial
in the value of the largest occurring number in every input.
\end{proof}

\subsection{Parameterized Algorithms}

Also the {\sc Max d-KP} problem 
is defined on inputs of various informations, which motivates us
to consider parameterized versions of the problem. By 
adding a threshold value $k$ for  the profit to the instance
and choosing a parameter $\kappa(I)$ from the instance
$I$, we define the following parameterized problem.

\begin{desctight}
\item[Name] $\kappa$-{\sc d-dimensional Knapsack} ($\kappa$-{\sc d-KP})

\item[Instance]  A set $A=\{a_1,\ldots,a_n\}$ of $n$ items and a number $d$ of dimensions.
Every item $a_j$ has a profit $p_j$ and for dimension $i$ the size $s_{i,j}$. 
Further for every dimension $i$ there is a capacity $c_i$ and we have given 
a positive integer $k$.

\item[Parameter] $\kappa(I)$

\item[Question] Is there a subset $A'\subseteq A$ such that
the total profit of $A'$  is at least $k$ 
and 
for every dimension $i$ the total size of $A'$ is at most the capacity $c_i$?
\end{desctight}

In Theorem 7 in \cite{KS10} it is shown that for $d=2$ dimensions 
$k$-{\sc d-KP} is $\w[1]$-hard. We generalize this result for $d\geq 2$
dimensions to obtain a result for parameter $(d,k)$ in Theorem 
\ref{maintheorem-d-kp}.

\begin{lemma}\label{le-d-kp}
For every $d\geq 2$ problem $k$-{\sc d-KP} is $\w[1]$-hard.
\end{lemma}

\begin{proof}
Let $I$ be an instance for $k$-{\sc d-KP} on $d$ dimensions and $n$ items 
of profits $p_j$ and sizes $s_{i,j}$ and $m$ knapsacks of capacities $c_i$. 
We define an instance $I'$ for $k$-{\sc d-KP}  on $d+1$ dimensions and $n$ items
of profits $p'_j$ and sizes $s'_{i,j}$ and $m$ knapsacks of capacities $c'_i$
as follows. The profits are $p'_j=p_j$ for $1\leq j \leq n$. The sizes are
$s'_{i,j}=s_{i,j}$ for $1\leq i\leq d$ and $1\leq j \leq n$ and $s'_{d+1,j}=1$ 
for $1\leq j \leq n$. The capacities are $c'_i=c_i$ for $1\leq i\leq d$ 
and $c_{d+1}=n$.
Then $I$ has a solution of profit $k$ if and only if $I'$ has
a solution of profit $k$.
\end{proof}

Next we give  parameterized algorithms for the d-dimensional knapsack problem.

\begin{table}[ht]
\begin{center}
\begin{tabular}{|l||l|lcl|}
\hline
Parameter       &  class & time    & & \\
\hline\hline
$n$            & $\not\in\pfpt$, $\in\fpt$   & $\bigo(d\cdot n \cdot 2^n)$ & $\subseteq$ & $\bigo(|I| \cdot 2^n)$ \\
\hline
$k$            &  $\w[1]$-hard, $\in\xp$   & $\bigo(d\cdot n^{k+1})$& $\subseteq$ & $\bigo(|I|\cdot n^{k})$     \\
\hline
 $(c_1,\ldots,c_d)$    &  $\in\pfpt$   & $\bigo(d \cdot n \cdot  \Pi_{i=1}^d c_i)$ & $\subseteq$ &$\bigo(|I| \cdot  \Pi_{i=1}^d c_i)$  \\
\hline
$(p_1,\ldots,p_n)$    & $\in\fpt$   & $\bigo(d\cdot n \cdot 2^{(\sum_{j=1}^{n}p_j)})$ & $\subseteq$ &  $\bigo(|I| \cdot 2^{(\sum_{j=1}^{n}p_j)})$     \\
\hline
$(s_{1,1},\ldots,s_{d,n})$ & $\in\fpt$   & $\bigo(d\cdot n \cdot 2^{\sum_{j=1}^n \sum_{i=1}^{d}s_{i,j}})$ & $\subseteq$ &  $\bigo(|I| \cdot 2^{\sum_{j=1}^n \sum_{i=1}^{d}s_{i,j}})$                                      \\

\hline
$d$            &    $\not\in\xp$ & & &\\
 \hline
$(d,\val)$  &  $\in\fpt$  & $\bigo(d\cdot n \cdot 2^{d\cdot\operatorname{val}(I)} )$  & $\subseteq$ & $\bigo( |I|\cdot 2^{d\cdot\operatorname{val}(I)} )$   \\
\hline

$(d,c_{\max})$  &  $\in\fpt$ &$\bigo(d\cdot n \cdot (c_{\max})^d)$ &$\subseteq$ & $\bigo(|I| \cdot (c_{\max})^d)$ \\
\hline

$(d,k)$  &   $\not\in\fpt$, $\in\xp$ &  $\bigo(d\cdot n^{k+1})$ & &   \\
\hline

\end{tabular}
\end{center}
\caption{Overview of parameterized algorithms for d-KP}\label{tab-dkp-survey}
\end{table}

\begin{theorem}\label{maintheorem-d-kp}
There exist parameterized algorithms for the d-dimensional knapsack problem
such that the running times of Table \ref{tab-dkp-survey} 
hold true and the problems belong to the specified parameterized complexity classes.
\end{theorem}

\begin{proof}
For parameter $\kappa(I)=n$ we can use
a brute force solution by checking for all $2^n$ possible subsets 
of $A$ within condition (\ref{lp-dkp2}) in the given  BIP, which
leads to an algorithm of time complexity $\bigo(d\cdot n \cdot 2^n)$.
Alternatively one can use  the result of \cite{Len83} or its improved
version in \cite{Kan87}, 
which implies that integer linear programming is fixed-parameter tractable for the 
parameter ''number of variables''. 
Thus by BIP (\ref{lp-dkp1})-(\ref{lp-dkp3}) problem $n$-{\sc d-KP}
is fixed-parameter tractable.

As a lower bound Corollary \ref{cor-pfpt-p} implies that there is no 
polynomial fpt-algorithm for  $n$-{\sc d-KP}.

For the standard parameterization $\kappa(I)=k$ we can use
the reduction given in the proof of Theorem \ref{th-multi-pseudo} in
order to obtain a parameterized reduction from the $k$-{\sc Independent Set} 
problem, which is $\w[1]$-hard, 
see \cite{DF13}. Thus $k$-{\sc d-KP} is $\w[1]$-hard.

In order to obtain an xp-algorithm that solves $k$-{\sc d-KP}
we apply the result of Lemma \ref{le-bd-sol-size-dkp} which allows
us to assume that every solution $A'$ of $k$-{\sc d-KP} has at most $k$
items. Thus we have to  check at most $n^k$ different possible solutions.
Each such solution can be verified with condition (\ref{lp-dkp2}) of the given 
BIP in time $\bigo(d\cdot n)$, which 
implies an xp-algorithm w.r.t.\ parameter $k$ of time $\bigo(d\cdot n^{k+1})$.

For parameter $\kappa(I)=(c_1,\ldots,c_d)$ we consider the dynamic programming 
algorithm cited in 
Theorem \ref{t-dy3} which has running time $\bigo(d \cdot n \cdot  \Pi_{i=1}^d c_i)$. 
From a parameterized point of view a polynomial fpt-algorithm that solves
$(c_1,\ldots,c_d)$-{\sc d-KP}  follows.

For parameter $\kappa(I)=(p_1,\ldots,p_n)$ we know that $n\leq \sum_{j=1}^{n}p_j$
since all profits are positive integers and thus we conclude an
fpt-algorithm for {\sc d-KP}  w.r.t. $(p_1,\ldots,p_n)$ by the result for parameter $n$.

An xp-algorithm w.r.t. $(p_1,\ldots,p_n)$ can be obtained as follows.
If for some $(p_1,\ldots,p_n)$-{\sc d-KP} instance it holds 
$\sum_{j=1}^{n}p_j < k$, then it is not possible to reach a profit of $k$.
Otherwise we know that $k\leq \sum_{j=1}^{n}p_j$ which implies by the upper result 
for parameter $k$ an xp-algorithm of running time  $\bigo(d\cdot n^{\sum_{j=1}^n p_j+1})$ 
for  $(p_1,\ldots,p_n)$-{\sc d-KP}.

The parameter $\kappa(I)=(s_{1,1},\ldots,s_{d,n})$ can be treated in the same 
way. Since all sizes are non-negative integers and $\sum_{i=1}^{d}s_{i,j}\geq 1$ for every item $a_j$
by our assumptions, it holds $n\leq \sum_{j=1}^n \sum_{i=1}^{d}s_{i,j}$
and thus we conclude an
fpt-algorithm for {\sc d-KP}  w.r.t. $(s_{1,1},\ldots,s_{d,n})$ by the result for parameter $n$.

If we choose $\kappa(I)=d$ then the parameterized problem is at 
least  $\w[1]$-hard, unless $\p=\np$. An fpt-algorithm with respect
to parameter $d$ would imply a polynomial time algorithm
for every fixed $d$, but even for $d=1$ the problem is NP-hard. 
For the same reason there is no xp-algorithm with respect
to parameter $d$.


Next we consider several combined parameters including $d$.
For parameter $\kappa(I)=(d,\val)$ we apply Theorem \ref{t-dy3} to obtain 
a parameterized running time of
$\bigo(d \cdot n \cdot (c_{\max})^d) \subseteq \bigo(d\cdot n \cdot (2^{\operatorname{val}(I)})^d)$.

For parameter $\kappa(I)=(d,c_{\max})$ we also can use
Theorem \ref{t-dy3} which leads to the fact that {\sc d-KP} can be solved in  time 
$\bigo(n \cdot d \cdot (c_{\max})^d)$.

For parameter $\kappa(I)=(d,k)$ the problem {\sc d-KP} is $\w[1]$-hard.
If $(d,k)$-{\sc d-KP} would be in $\fpt$, then for every fixed dimension $d$ 
problem $k$-{\sc d-KP} is fixed-parameter tractable in contradiction to
Lemma \ref{le-d-kp}. Since the running time $\bigo(d\cdot n^{k+1})$, which was
mentioned above for parameter $k$, is polynomial for every
fixed $k$ and $d$ the problem $(d,k)$-{\sc d-KP}  is slicewise polynomial.
\end{proof}

Theorem \ref{maintheorem-d-kp} states $\w[1]$-hardness of $k$-{\sc d-KP}. 
But this changes if we only look for solutions of high profit.
Under a restriction $n/c\leq k$, for some constant $c>1$, fixed
parameter tractability with regard to $n$ then implies the problem to be in
$\fpt$ with respect to parameter $k$.

\subsection{Kernelizations}

Next we give kernelization bounds for the d-dimensional knapsack problem.

\begin{table}[ht]
\begin{center}
\begin{tabular}{|l||l|l|l|}
\hline
Parameter      & lower bound  & upper bound  \\
\hline\hline
$n$    ~~~~~~~~   if $d\leq n$     & $\omega(1)$  & $\bigo(n^5)$  \\
\hline
 $(c_1,\ldots,c_d)$     & $\Theta(1)$ & $\Theta(1)$\\
\hline
$(p_1,\ldots,p_n)$    & &   $\bigo(2^{(\sum_{j=1}^{n}p_j)})$ \\
\hline
$(s_{1,1},\ldots,s_{d,n})$  &  & $\bigo(2^{\sum_{j=1}^n \sum_{i=1}^{d}s_{i,j}})$\\
\hline
$(d,\val)$         & &  $\bigo(2^{d\cdot\operatorname{val}(I)} )$  \\
\hline
$(d,c_{\max})$     &  &$\bigo((c_{\max})^d)$  \\
\hline
\end{tabular}
\end{center}
\caption{Overview for kernel sizes of parameterized d-KP}\label{tab-dkp-ker-survey}
\end{table}

\begin{theorem}\label{maintheorem2-dkp}
There exist kernelizations for the parameterized d-dimensional knapsack problem 
such that the upper bounds for the sizes
of a possible kernel in Table \ref{tab-dkp-ker-survey} 
hold true.
\end{theorem}

\begin{proof}
For parameter  $\kappa(I)=n$ let $I$ be an instance of $n$-{\sc d-KP}.
We proceed as in the proof of Theorem \ref{maintheorem2}.
In the case of $n$-{\sc d-KP} we have to scale $d$  
inequalities of type (\ref{eq-ker-n-1}) and one inequality of type (\ref{eq-ker-n-2})
by Theorem \ref{th-ft87}.
For the obtained instance $I'$ we can
bound $|I'|$ by the number of items $n$ and $d(n+1)+(n+1)=(d+1)(n+1)$ numbers
of value at most $2^{4(n+1)^3} (\ell+2)^{(n+1)(n+3)}$.
Since we can assume $\ell\leq n$ for {\sc
d-KP}, this problem has a  kernel of size 
$$\bigo\left(n+ (d+1)(n+1) \log_2\left( 2^{4(n+1)^3} (\ell+2)^{(n+1)(n+3)} \right) \right)
\subseteq \bigo(d\cdot n^4 + d\cdot \log_2(\ell+2)n^3) \subseteq \bigo(d\cdot n^4),
$$
which implies a kernel for $n$-{\sc d-KP} of size $\bigo(n^4)$ for every fixed $d$ and
a kernel of size $\bigo(n^5)$ for  $d\leq n$.

For parameter $\kappa(I)=(c_1,\ldots,c_d)$ we obtain by 
Theorem \ref{maintheorem-d-kp} and Theorem \ref{th-me}  a kernel of constant
size.

For the remaining four parameters of Table \ref{tab-dkp-ker-survey}  the upper bounds follow by
Theorem \ref{maintheorem-d-kp} and Theorem \ref{th-ker}.
\end{proof}

\section{Multiple Knapsack Problem}\label{def-sec-mkp}

Next we consider the multiple knapsack problem.

\begin{desctight}
\item[Name] {\sc Max multiple Knapsack} ({\sc Max MKP})

\item[Instance] A set $A=\{a_1,\ldots,a_n\}$ of $n$ items and a number $m$ of knapsacks.
Every item $a_j$ has a profit $p_j$ and a size $s_j$.
Each knapsack $i$  has a capacity $c_i$. 

\item[Task] Find $m$ disjoint (possibly empty) subsets $A_1,\ldots, A_m$ of $A$ such that 
the total profit of the selected items  is maximized 
and 
each subset can be assigned to a different knapsack 
$i$ without exceeding its capacity $c_i$ by the sizes of the selected items.
\end{desctight}

For $m=1$ the {\sc Max MKP} problem corresponds to the {\sc Max KP} problem considered in
Section  \ref{sec-kp}.
{\sc Max MKP} does not allow an FPTAS even for $m=2$ knapsacks, see \cite{CK00} or \cite{CKP00}. 
{\sc Max MKP} allows an EPTAS of running time 
$2^{\bigo(\nicefrac{1}{\epsilon}\cdot\log^4(\nicefrac{1}{\epsilon}))} + n^{\bigo(1)}$, see \cite{Jan12}.

The parameters $n$, $m$, $p_j$, $s_{j}$, and $c_i$  are assumed to be positive integers. 
Let $s_{\max}=\max_{1\leq j \leq n}s_j$, 
$p_{\max}=\max_{1\leq j \leq n}p_{j}$, and
$c_{\max}=\max_{1\leq i \leq m}c_i$. 
The same notations are also used for $\min$ instead of $\max$.
In order to avoid trivial solutions
we assume that $s_{\max}\leq c_{\max}$, $s_{\min}\leq c_{\min}$, and  $\sum_{j=1}^{n} s_j > c_{\max}$.
Further we can assume that $n\geq m$, since otherwise we can eliminate the $m-n$ knapsacks
of smallest capacity.

For some instance $I$ of {\sc Max MKP} its size $|I|$ can be bounded (cf.~\cite{GJ79}) by 
$$
\begin{array}{lcl}
|I| &  =  & n + \sum_{j=1}^{n}(1+\lfloor\log_2(p_j)\rfloor) + \sum_{j=1}^{n}(1+\lfloor\log_2(s_j)\rfloor)+  \sum_{i=1}^{m}(1+\lfloor\log_2(c_i)\rfloor)\\
    & \in & \bigo(n + m + \sum_{j=1}^{n}\log_2(p_j) + \sum_{j=1}^{n}\log_2(s_j)+  \sum_{i=1}^{m}\log_2(c_i))\\
    & = & \bigo(n + m + n\cdot \log_2(p_{\max}) + n\cdot \log_2(s_{\max})+ m\cdot \log_2(c_{\max})).\\
\end{array}
$$

By assuming that we have one knapsack of capacity $\sum_{i=1}^m c_i$
similar as in the poof of Lemma \ref{le-blb} we can bound the
number of items w.r.t.\ the sum of all capacities. 

\begin{lemma}[Reduced Instance]\label{le-red-p}
Every  instance of {\sc Max MKP}  can be 
transformed into a reduced instance, such that  
$n\in \bigo(\log (c_{\max})\cdot  \sum_{i=1}^m c_i)$.
\end{lemma}

\begin{proof}
First we can assume, that there is no item
in $A$, whose size is larger than the maximum capacity $c_{\max}$, i.e.\ 
$s_j\leq c_{\max}$ for every $1\leq j \leq n$.
Further for $1 \leq s\leq c_{\max}$
we can assume that there are at most $n_s:=\lfloor\frac{\sum_{i=1}^m c_i}{s}\rfloor$  
items of size $s$ in $A$. 

By the harmonic series we always can bound the number $n$ of items in $A$ by
$$\begin{array}{lcl}
n&\leq &\sum_{s=1}^{c_{\max}} n_s \\
&=&  \sum_{s=1}^{c_{\max}}  \bigg\lfloor\frac{\sum_{i=1}^m c_i}{s}\bigg\rfloor \\
&\leq &\sum_{s=1}^{c_{\max}} \frac{\sum_{i=1}^m c_i}{s} \\
&=& (\sum_{i=1}^m c_i) \cdot \sum_{s=1}^{c_{\max}} \frac{1}{s}  \\
&< & (\sum_{i=1}^m c_i) \cdot (\ln(c_{\max})+1) \\
&\in &\bigo((\sum_{i=1}^m c_i) \cdot \log (c_{\max})).
\end{array}$$

If we have given an instance $I_1$ for {\sc Max MKP}  with more than 
the mentioned number $n_s$ of items of size  $s$ for some $1 \leq s\leq c_{\max}$, we remove
all of them except the $n_s$ items of the highest profit. 
The new instance $I_2$ satisfies $n\in \bigo(\log (c_{\max})\cdot  \sum_{i=1}^m c_i)$ and 
is a reduced instance of $I_1$. 
\end{proof}

\begin{lemma}[Bounding the size of a solution]\label{le-bd-sol-size} 
For every instance of {\sc Max MKP}
there is a feasible solution $A_1,\ldots,A_m$ of profit at least $k$ 
if and only if
there is a feasible solution $A'_1,\ldots,A'_m$ of profit at least $k$, 
which has size at most $k$.
\end{lemma}

\begin{proof}
If $|A_1\cup \ldots \cup A_m|\geq k+1$ we can remove one item of smallest 
profit, since all profits are  positive integers. 
\end{proof}

\subsection{Binary Integer Programming}

By choosing  a boolean variable $x_{i,j}$ for every item $a_j\in A$ and every
knapsack $1\leq i\leq m$,
indicating whether or not the item $a_j$ will be put into knapsack $i$,
a binary integer programming (BIP) version of the {\sc Max MKP} problem is 
as follows.
\begin{eqnarray}
&\text{max } &  \sum_{i=1}^{m} \sum_{j=1}^{n}  p_j \cdot x_{i,j}                          \label{lp-mkp1}   \\
&\text{s.t. }&  \sum_{j=1}^{n} s_{j} \cdot x_{i,j}  \leq    c_i   \text{ for } i\in [m]   \label{lp-mkp2} \\
&\text{and  }&  \sum_{i=1}^{m} x_{i,j}  \leq  1   \text{ for } j\in [n]                   \label{lp-mkp3} \\
&\text{and  }&  x_{i,j}\in\{0,1\}  \text{ for }  i\in [m], j\in [n]                       \label{lp-mkp4} 
\end{eqnarray}

The condition (\ref{lp-mkp3}) ensures that all knapsacks are disjoint, i.e.\
every item is contained in at most one knapsack.

\subsection{Dynamic Programming Algorithms}

A dynamic programming  solution for {\sc Max MKP} 
for $m=2$ knapsacks can be found in \cite{PT99}, 
which can be generalized as follows.

\begin{theorem}\label{t-dy3x}
{\sc Max MKP} can be solved in time 
$\bigo(n\cdot m \cdot \prod_{i=1}^{m}c_i)\subseteq \bigo(n\cdot m \cdot (c_{\max})^m)$
\end{theorem}

\begin{proof}
We define $P[k,c'_1,\ldots,c'_m]$ to be the maximum profit of the subproblem
where we only may choose a subset from the first $k$ items
$a_1,\ldots,a_k$ and the capacities are $c'_1,\ldots,c'_m$.
We initialize $P[0,c'_1,\ldots,c'_m]=0$ for all  $c'_1,\ldots,c'_m$ since
when choosing none of the items, the profit is always zero.
Further we set $P[k,c'_1,\ldots,c'_m]=-\infty$ if at least one of the $c'_1,\ldots,c'_m$
is negative in order to represent the case where the size $s_k$ of an item
is too high for packing it into a knapsack of capacity $c'_i$.
The values $P[k,c'_1,\ldots,c'_m]$, $1 \leq k\leq n$, 
for every $0\leq c'_i\leq c_i$ and every 
$1\leq i\leq m$ can be computed by the following recursion.
\[
P[k,c'_1,\ldots,c'_m]    = \max \left\{
 \begin{array}{lll}
       P[k-1,c'_1,\ldots,c'_m] \\ 
        P[k-1,c'_1-s_k,\ldots,c'_m]+p_k \\ 
     $\vdots$ \\
P[k-1,c'_1,\ldots,c'_m-s_k]+p_k \\ 
     \end{array}
   \right.
\]
All these values define a table with $\bigo(n\cdot \prod_{i=1}^{m}c_i)$
fields, where each field can be computed in time $\bigo(m)$.
The optimal return is $P[n,c_1,\ldots,c_m]$. Thus we have shown that 
{\sc Max MKP} can be solved in time
$\bigo(n\cdot  m \cdot \prod_{i=1}^{m}c_i)$. 
\end{proof}


\subsection{Pseudo-polynomial Algorithms} \label{multip-pseudo2}

The existence of pseudo-polynomial algorithms
for {\sc Max MKP} depends
on the assumption whether the number of knapsacks $m$ is given in the input
or is assumed to be fixed.

\begin{theorem}\label{th-m-pseudo}
{\sc Max MKP} is not pseudo-poly\-no\-mial.
\end{theorem}

\begin{proof}
We give a pseudo-polynomial reduction (cf.~page 101 in \cite{GJ79}) from
{\sc 3-Partition} which is not pseudo-poly\-no\-mial by \cite{GJ79}. 
Given are $n=3m$ positive integers $w_1,\ldots, w_n$ such that the 
sum $\frac{1}{m}\cdot \sum_{j=1}^n w_j=B$ and
$\nicefrac{B}{4}< w_j < \nicefrac{B}{2}$ for every $1\leq j \leq n$.
The question is to decide whether there is a partition of $N=\{1,\ldots,n\}$ 
into $m$ sets $N_1,\ldots,N_m$ such that $\sum_{j\in N_i}w_j=B$ for 
every $1\leq i \leq m$.

Let $I$ be an instance for the {\sc 3-Partition}. 
We define an instance $I'$ for {\sc Max MKP} by choosing
the number of items as $n$, the number of knapsacks as $m$, 
the capacities\footnote{The choice of the capacities in the proof 
of Theorem \ref{th-m-pseudo} shows that even {\sc Multiple Knapsack
with identical capacities (MKP-I)}, see \cite{KPP10}, is not pseudo-poly\-no\-mial.} 
$c_i=B$ for $1\leq i \leq m$, the profits $p_j=1$ and sizes $s_j=w_j$ for  
$1\leq j \leq n$.
By this construction every  {\sc 3-Partition} solution for $I$ implies a 
solution with optimal profit $n$ for the {\sc Max MKP} 
instance $I'$ and vice versa.
\end{proof}

\begin{theorem}\label{th-m-pseudo2}
For every fixed $m$ there is a pseudo-polynomial algorithm that solves 
{\sc Max MKP} in time $\bigo(n\cdot m \cdot (c_{\max})^m)$.
\end{theorem}

\begin{proof}
We consider the algorithm with running 
time $\bigo(n\cdot m \cdot (c_{\max})^m)$  of 
Theorem \ref{t-dy3x}. If $m$ is assumed to be fixed, $n\cdot m$
is polynomial in the input size and $(c_{\max})^m$ is polynomial
in the value of the largest occurring number in every input.
\end{proof}

\subsection{Parameterized Algorithms}

Also the {\sc Max MKP} problem 
is defined on inputs of various informations, which motivates us
to consider parameterized versions of the problem. By 
adding a threshold value $k$ for  the profit to the instance
and choosing a parameter $\kappa(I)$ from the instance
$I$, we define the following parameterized problem.

\begin{desctight}
\item[Name] $\kappa$-{\sc multiple Knapsack} ($\kappa$-{\sc MKP})

\item[Instance] A set $A=\{a_1,\ldots,a_n\}$ of $n$ items and a number $m$ of knapsacks.
Every item $a_j$ has a profit $p_j$ and a size $s_j$.
Each knapsack $i$  has a capacity $c_i$ and we have given 
a positive integer $k$.

\item[Parameter] $\kappa(I)$

\item[Question] Are there $m$ disjoint (possibly empty) subsets 
$A_1,\ldots, A_m$ of $A$ such that  
the total profit of the selected items is at least $k$
and 
each subset can be assigned to a different knapsack 
$i$ without exceeding its capacity $c_i$ by the sizes of the selected items?
\end{desctight}

We give a bound on the number of items in the threshold value for the profit $k$.

\begin{lemma}[Reduced Instance]\label{le-blb2}
Every instance of {\sc Max MKP} 
can be transformed into a reduced instance, such that  
$n\in \bigo(k\cdot \log(k))$.
\end{lemma}

\begin{proof}
For some fixed profit $p$, $p\geq 1$ we need at most $\lceil\frac{k}{p}\rceil$ many items
of profit $p$ in order to reach the profit $k$. 
If we have given an instance for {\sc Max MKP}  with more than $\lceil\frac{k}{p}\rceil$
items of profit $p$, we remove
all of them except the $\lceil\frac{k}{p}\rceil$ items of the smallest size. 
Since each profit is a positive integer we can  bound the number of items by
$n\leq \sum_{p=1}^{\infty}  \lceil\frac{k}{p}\rceil$.
Further all items with profit $p\geq k$ can be replaced by one of them 
of the smallest size. Thus we can assume that
$n\leq \sum_{p=1}^{k}  \lceil\frac{k}{p}\rceil$.
Further by the harmonic series it holds
$$n\leq \sum_{p=1}^{k}  \bigg\lceil\frac{k}{p}\bigg\rceil \leq \sum_{p=1}^{k} \bigg(\frac{k}{p}+1\bigg) =  k+ \sum_{p=1}^{k} \frac{k}{p} =  k+ k\cdot \sum_{p=1}^{k} \frac{1}{p} <  k+ k\cdot ( \ln(k) +1) \in \bigo(k\cdot \log (k)).$$
\end{proof}

Next we give  parameterized algorithms for the multiple knapsack problem.
Therefor let $B(n)=\sum_{i=0}^{n-1} \binom{n-1}{i} B(i)$ be the $n$-th {\em Bell number} 
which asymptotically grows faster than 
$c^n$ for every constant $c$ but slower than $n$ factorial.

\begin{table}[ht]
\begin{center}
\begin{tabular}{|l||l|lcl|}
\hline
Parameter      &  class & time   & &      \\
\hline\hline
$n$            & $\not\in\pfpt$, $\in\fpt$   & $\bigo(B(n)\cdot (m \cdot \log(m)+n))$  & $\subseteq$ &$\bigo(|I|^2 \cdot B(n))$ \\
\hline
$k$            &  $\in\fpt$  &  $\bigo((m \cdot \log(m)+n)  \cdot 2^{k\cdot \ln^2(k)})$  &$\subseteq$ & $\bigo( |I|^{2} \cdot 2^{k\cdot \ln^2(k)})$ \\
\hline
$(c_1,\ldots,c_m)$&  $\in\pfpt$  &  $\bigo(n\cdot m \cdot \prod_{i=1}^{m}c_i)$   &$\subseteq$& $\bigo(|I|^2 \cdot \prod_{i=1}^{m}c_i)$   \\
\hline
$(p_1,\ldots,p_n)$&    $\in\fpt$ &   $\bigo(B(\sum_{j=1}^{n}p_j)\cdot (m \cdot \log(m)+n))$  &$\subseteq$&   $\bigo(|I|^2 \cdot B(\sum_{j=1}^{n}p_j))$   \\  

\hline
$(s_1,\ldots,s_n)$&    $\in\fpt$ &   $\bigo(B(\sum_{j=1}^{n}s_j)\cdot (m \cdot \log(m)+n))$    &$\subseteq$&  $\bigo(|I|^2 \cdot B(\sum_{j=1}^{n}s_j))$ \\  
\hline
$m$ &   $\not\in\xp$ &  & &\\
 \hline

$(m,\val)$ &  $\in\fpt$ & $\bigo( n\cdot m \cdot  2^{m\cdot \operatorname{val}(I)})$    &$\subseteq$ & $\bigo( |I|^2 \cdot  2^{m\cdot \operatorname{val}(I)})$   \\
\hline
$(m,c_{\max})$ &  $\in\fpt$ & $\bigo(n\cdot m \cdot (c_{\max})^m)$ &$\subseteq$ & $\bigo(|I|^2 \cdot (c_{\max})^m)$  \\
\hline
$(m,n)$ &  $\in\fpt$  &$\bigo(n\cdot m \cdot 2^{n\cdot m})$ & $\subseteq$ &$\bigo(|I|^2 \cdot 2^{n\cdot m})$ \\
\hline
\end{tabular}
\end{center}
\caption{Overview for parameterized algorithms for MKP}\label{tab-mkp-survey}
\end{table}

\begin{theorem}\label{maintheorem-mkp}
There exist parameterized algorithms for the multiple knapsack problem
such that the running times of Table \ref{tab-mkp-survey} 
hold true and the problems belong to the specified parameterized complexity classes.
\end{theorem}

\begin{proof}
We start with parameter $\kappa(I)=n$. 
A brute force solution to solve $n$-{\sc MKP} is to consider all $B(n)$ possible partitions of the $n$ 
items into disjoint non-empty subsets, which all can be generated in time $\bigo(B(n))$ 
by \cite{KN05a}.
For every such partition we check if it consists of at most $m$ sets 
and after sorting descending w.r.t.\ the sum of sizes in each set $A'_j$ of the 
partition, we compare this sum with the capacity $c_j$. The capacities are
also assumed to be sorted, such that $c_1\geq c_2 \geq \ldots \geq c_m$.
Every of these partitions can be handled in time $\bigo(m\cdot \log(m)+n)$
for sorting the sets $A'_j$  and summing up the sizes.
This leads to an fpt-algorithm that solves $n$-{\sc MKP}
in time $\bigo(B(n)\cdot (m \cdot \log(m)+n))$.

As a lower bound Corollary \ref{cor-pfpt-p} implies that there is no 
polynomial fpt-algorithm for  $n$-{\sc MKP}.

For parameter $\kappa(I)=k$ by the result for parameter $n$ 
and  Lemma \ref{le-blb2} we obtain an fpt-algorithm that solves $k$-{\sc MKP}
in time $\bigo(B(k\cdot \ln(k))\cdot (m \cdot \log(m)+n))$.
By the bound on the Bell number $B(\ell)\leq (\frac{\ell}{\ln(\ell)})^\ell$ shown in \cite{BT10} we obtain the following inclusions.
$$\begin{array}{lcl}
\bigo((\frac{k\cdot \ln(k)}{\ln(k\cdot \ln(k))})^{k\cdot \ln(k)}\cdot (m \cdot \log(m)+n))&\subseteq& \bigo(k^{k\cdot \ln(k)}\cdot (m \cdot \log(m)+n))\\
&\subseteq & \bigo(2^{k\cdot \ln^2(k)}\cdot (m \cdot \log(m)+n))
\end{array}$$
Thus there is an fpt-algorithm that solves $k$-{\sc MKP}
in time $\bigo(2^{k\cdot \ln^2(k)}\cdot (m \cdot \log(m)+n))  \subseteq \bigo(2^{k\cdot \ln^2(k)}\cdot n^{2})$.

For parameter $\kappa(I)=(c_1,\ldots,c_m)$ we use the running time
mentioned in Theorem \ref{t-dy3x} to obtain
a polynomial fpt-algorithm that solves $(c_1,\ldots,c_m)$-MKP
in time $\bigo(n\cdot m \cdot \prod_{i=1}^{m}c_i)$.

For parameter $\kappa(I)=(p_1,\ldots,p_n)$ we know that $n\leq \sum_{j=1}^{n}p_j$
since all profits are positive integers and thus we conclude an
fpt-algorithm of running time
 $\bigo(B(\sum_{j=1}^{n}p_j)\cdot (m \cdot \log(m)+n))$
for {\sc MKP}  w.r.t. $(p_1,\ldots,p_n)$ by the result for parameter $n$.

For parameter $\kappa(I)=(s_1,\ldots,s_n)$ we know that $n\leq \sum_{j=1}^{n}s_j$
since all sizes are positive integers and in the same way as for
the profits  we conclude an fpt-algorithm for {\sc MKP}  w.r.t. $(s_1,\ldots,s_n)$.

If we choose $\kappa(I)=m$ the parameterized problem is at 
least  $\w[1]$-hard, unless $\p=\np$. An fpt-algorithm with respect
to parameter $m$ would imply a polynomial time algorithm
for every fixed $m$, but even for $m=1$ the problem is NP-hard. 
For the same reason there is no xp-algorithm with respect
to parameter $m$.

Next we consider several combined parameters including $m$.
For parameter $\kappa(I)=(m,\val)$ we 
apply Theorem \ref{t-dy3x} to obtain 
a parameterized running time of
$\bigo(n \cdot m \cdot (c_{\max})^m) \subseteq \bigo(n\cdot m \cdot 
(2^{\operatorname{val}(I)})^m) \subseteq \bigo(n\cdot m \cdot 2^{m\cdot\operatorname{val}(I)})$.
Thus there is an fpt-algorithm that solves $(m,\val)$-{\sc MKP}
in time $\bigo(n\cdot m \cdot 2^{m\cdot\operatorname{val}(I)})$.

For parameter $\kappa(I)=(m,c_{\max})$ we also use
Theorem \ref{t-dy3x} to show that the $(m,c_{\max})$-{\sc MKP} 
problem is fixed-parameter tractable with running time 
$\bigo(n \cdot m \cdot (c_{\max})^m)$

Finally for the parameter $\kappa(I)=(m,n)$ we verify all possible $2^{n\cdot m}$ assignments
of the $n\cdot m$ variables in  (\ref{lp-mkp1})-(\ref{lp-mkp4})  each 
in time $\bigo(n\cdot m)$.
Thus there is an fpt-algorithm that solves  $(m,n)$-{\sc MKP} 
in time $\bigo(n\cdot m \cdot 2^{n\cdot m})$.

Alternatively we can apply
the result of \cite{Kan87}, which implies
that integer linear programming is fixed-parameter tractable for the 
parameter ''number of variables''. By BIP (\ref{lp-mkp1})-(\ref{lp-mkp4})  
for {\sc Max MKP} we obtain an fpt-algorithm that solves  $(m,n)$-{\sc MKP} 
in time $\bigo((n\cdot m)^{\bigo(n\cdot m)}\cdot |I|)$.
\end{proof}

Next we give some further ideas how to show parameterized algorithms for  $k$-{\sc MKP}.
Although these are less efficient, the ideas might be interesting on its own.

\begin{remark}
\begin{enumerate}
\item The {\sc Max MKP} problem  allows an 
EPTAS of parameterized running time 
$2^{\bigo(\nicefrac{1}{\epsilon}\cdot\log^4(\nicefrac{1}{\epsilon}))} + n^{\bigo(1)}$, 
see  \cite{Jan12}. By Theorem \ref{theorem-fg}  we can use
this EPTAS for $\epsilon=\nicefrac{1}{2k}$ in order to obtain an fpt-algorithm that solves 
the standard parameterization of the corresponding decision problem 
in time $2^{\bigo(k\cdot\log^4(k))}+ n^{\bigo(1)}$. 
Thus there is an fpt-algorithm that solves $k$-{\sc MKP}
in time $2^{\bigo(k\cdot\log^4(k))}+ n^{\bigo(1)}$.

\item
By Lemma \ref{le-bd-sol-size} we can restrict to solutions
which have size at most $k$. Formally we need to bound the 
number 
of families of disjoint subsets $A'_1, \ldots, A'_m$ 
of a set $A$ on $n$ elements, such that $|A'_1\cup \ldots \cup A'_m|\leq k$.
Thus for $1\leq i\leq k$ we sum up $B(i)$ multiplicated by 
$\binom{n}{i}$ possible ways to choose the $i$ items from all
$n$ items.
$$
\begin{array}{lcll}
\sum_{i=1}^{k}\binom{n}{i}B(i)&  =   & \sum_{i=1}^k \binom{k}{i} \frac{(k-i)! n!}{k! (n-i)!} B(i)  & \text{since } \frac{\binom{n}{i}}{\binom{k}{i}}= \frac{(k-i)! n!}{k! (n-i)!}\\
                              & \leq & \frac{n!}{(n-k)!}  \sum_{i=1}^k \binom{k}{i} B(i)           & \text{since } i\leq k\\
                              & \leq & \frac{n!}{(n-k)!} B(k+1)                                    & \text{by  definition of } B(k+1)\\
                              & \leq & n^k B(k+1)                                                  & \text{since } \frac{n!}{(n-k)!}\leq n^k       
\end{array}$$
Every of these partitions can be generated in constant time by \cite{KN05a} and  
handled in time $\bigo(m\cdot \log(m)+n)$ as mentioned in the proof of Theorem \ref{maintheorem-mkp}.
Thus  there is an xp-algorithm that solves $k$-{\sc MKP} 
in time $\bigo(n^k \cdot B(k+1) \cdot (m \cdot \log(m)+n)) \subseteq \bigo(n^{k+2} \cdot B(k+1))$.
\end{enumerate}
\end{remark}

\subsection{Kernelizations}

Next we give kernelization bounds for the multiple knapsack problem.

\begin{table}[ht]
\begin{center}
\begin{tabular}{|l||l|l|l|}
\hline
Parameter               & lower bound  & upper bound  \\
\hline\hline
$n$                     & $\omega(1)$  & $\bigo(n^{8})$  \\
\hline
$k$                     &              & $\bigo(2^{k\cdot \ln^2(k)})$ \\
\hline
$(c_1,\ldots,c_m)$      & $\Theta(1)$  & $\Theta(1)$\\
\hline
$(p_1,\ldots,p_n)$      &              & $\bigo(B(\sum_{j=1}^n p_j))$ \\
\hline
$(s_{1},\ldots,s_{n})$  &              & $\bigo(B(\sum_{j=1}^n s_j))$\\
\hline
$(m,\val)$              &              & $\bigo(2^{m\cdot\operatorname{val}(I)} )$  \\
\hline
$(m,c_{\max})$          &              & $\bigo((c_{\max})^m)$  \\
\hline
$(m,n)$                 &              & $\bigo(2^{n\cdot m})$  \\
\hline
\end{tabular}
\end{center}
\caption{Overview for kernel sizes of parameterized MKP}\label{tab-mkp-ker-survey}
\end{table}

\begin{theorem}\label{maintheorem2-mkp}
There exist kernelizations for the parameterized multiple knapsack problem 
such that the upper bounds for the sizes
of a possible kernel in Table \ref{tab-mkp-ker-survey} hold true.
\end{theorem}

\begin{proof}
For parameter  $\kappa(I)=n$
we proceed as in the proof of Theorem \ref{maintheorem2} which
shows  a kernel  of size $\bigo(n^4)$ for $n$-{\sc KP}.
In the case of $n$-{\sc MKP} we have to scale $m$  
inequalities of type (\ref{eq-ker-n-1}) on $n$ 
variables and one inequality of type (\ref{eq-ker-n-2})
on $n\cdot m$ variables
by Theorem \ref{th-ft87}.
For the obtained instance $I'$ we can
bound $|I'|$ by the number of items $n$ and $m(n+1)$
numbers of value at most $2^{4(n+1)^3} (\ell+2)^{(n+1)(n+3)}$
and $n\cdot m +1$ numbers of value at most
 $2^{4(nm+1)^3} (\ell+2)^{(nm+1)(nm+3)}$.
Thus  $n$-{\sc MKP} has a  kernel of size 
$$\bigo\left(n+ m(n+1) \log_2( 2^{4(n+1)^3} (\ell+2)^{(n+1)(n+3)})  + (nm +1) \log_2( 2^{4(nm+1)^3} (\ell+2)^{(nm+1)(nm+3)})  \right).$$
We can assume $m\leq n$ (cf.~beginning of Section \ref{def-sec-mkp}) and 
since every item is assigned to at most one knapsack  (cf.~(\ref{lp-mkp3})) we know
$\ell \leq n$. Thus we obtain a kernel of size 
$$
\bigo(n\cdot m \cdot n^3 + n\cdot m \cdot \log_2(\ell+2)n^2 + n\cdot m \cdot (nm)^3  + n\cdot m \cdot\log_2(\ell+2)(nm)^2) \subseteq \bigo(n^{4}\cdot m^4) \subseteq\bigo(n^{8}).
$$
for $n$-{\sc MKP}.

For parameter $\kappa(I)=(c_1,\ldots,c_m)$ we obtain by 
Theorem \ref{maintheorem-mkp} and Theorem \ref{th-me}  a kernel of constant
size.

For the remaining six parameters of Table \ref{tab-mkp-ker-survey}  the upper bounds follow by
Theorem \ref{maintheorem-mkp} and Theorem \ref{th-ker}.
\end{proof}

\section{Conclusions and Outlook}\label{sec-con}

We have considered the {\sc Max Knapsack}
problem and its two generalizations {\sc Max Multidimensional Knapsack} 
and {\sc Max Multiple Knapsack}. 
The parameterized decision versions
of all three problems allow  several parameterized algorithms.


From a practical point of view choosing the standard parameterization $k$  is
not very useful, since a large profit of the subset $A'$
violates the aim that a good parameterization is small for every
input. So for KP we suggest it is better to choose the capacity as a parameter,
i.e.\ $\kappa(I)=c$, since common values of $c$ are low enough such that
the polynomial fpt-algorithm is practical. The same holds for d-KP and MKP.
Further one has a good parameter, if it is smaller than the input size $|I|$ but measures
the structure of the instance. This is the case for the parameter number
of items $n$ within all three considered knapsack problems.

The special case of the {\sc Max Knapsack} problem, 
where $s_j=p_j$ for all items $1\leq j \leq n$ is known 
as the {\sc Subset Sum} problem. For this case we know that 
$s_{\max}=p_{\max}\leq c$ and we conclude the existence of 
fpt-algorithms with respect to parameter $n$, $c$, and $k$.
Kernels for the  {\sc Subset Sum} problem w.r.t. $n$ and the number
of different sizes $\sizevar$ are examined in
\cite{EKMR15}.

The closely related minimization problem 
\begin{eqnarray}
\text{min }   \sum_{i=1}^{n}  x_i  \text{~~ s.t. ~~}  \sum_{i=1}^{n} c_{i}
 x_{i}  =    c  \text{~~ and ~~} x_{i}\in \{0,1\}  \text{ for } i\in [n]  \label{lp}
\end{eqnarray}
is known as the {\sc Change Making} problem, whose parameterized complexity
is discussed in \cite{GGRY15}.

In our future work, we want to find better fpt-algorithms, especially 
for {\sc d-KP} and {\sc MKP}. 
We also want to consider the following additional parameters.
\begin{itemize}
\item $\pvar(I)=|\{p_1,\ldots,p_n\}|$ for KP
\item $\pvar(I)=|\{p_1,\ldots,p_n\}|$, $\sizevar(I)=|\{s_{1,1},\ldots,s_{d,n}\}|$, 
$c_{\max}$,  $p_{\max}$,  $s_{\max}$, and $\val$ for d-KP, and
\item $\sizevar(I)=|\{s_1,\ldots,s_n\}|$, $\pvar(I)=|\{p_1,\ldots,p_n\}|$, $\val$,
$c_{\max}$,  $p_{\max}$, and $s_{\max}$ for MKP.
\end{itemize}

Also from a theoretical point of view it is interesting to increase
the number of parameters for which the parameterized complexity 
of the considered problems is known. 
For example if our problem is $\w[1]$-hard with respect to some parameter $\kappa$, 
then a natural question is to ask, whether it remains hard for the {\em dual parameter} 
$\kappa_d(I)=\max_{I'\in {\cal I}}\kappa(I') -\kappa(I)$. 
That is, if $\kappa$ measures the costs of a solution, then for some optimization problem
the dual parameter $\kappa_d$ measures the costs of the elements that are not in the solution \cite{Nie10}.
Since $k$-{\sc d-KP} is $\w[1]$-hard the question arises, whether {\sc d-KP} becomes tractable
w.r.t.\ parameter $n\cdot p_{\max}-k$. 
More general, one also might consider more combined parameters, i.e.\
parameters that consists of two or more parts of the input. For {\sc d-KP}
combined parameters including $k$ are of our interest.

The existence of polynomial kernels for knapsack problems seems to be
nearly uninvestigated. Recently  
a polynomial kernel for {\sc KP} using rational sizes and profits 
is constructed in \cite{EKMR15,EKMR16} by Theorem  \ref{th-ft87}.
This result also holds for integer sizes and profits  (cf. Theorem \ref{maintheorem2}).
By considering polynomial fpt-algorithms we could show some lower bounds for
kernels for KP (cf. Table \ref{tab-kp-ker-survey}). 
We want to consider further kernels for {\sc d-KP} and {\sc MKP}, try to improve 
the sizes of known kernels, and give lower bounds for the
sizes of kernels.

A further task is to extend the results to more knapsack problems, 
e.g. {\sc Max-Min Knapsack} 
problem and restricted versions of the presented problems, e.g. {\sc Multiple Knapsack
with identical capacities (MKP-I)}, see \cite{KPP10}.

We also want to consider the existence of parameterized approximation algorithms
for knapsack problems, see \cite{Mar08} for a survey.

\section*{Acknowledgements}

We would like to thank Klaus Jansen and Steffen Goebbels for useful discussions.


\begin{thebibliography}{10}

\bibitem{ACGKMP99}
G.~Ausiello, P.~Crescenzi, G.~Gambosi, V.~Kann, A.~Marchetti-Spaccamela, and
  M.~Protasi.
\newblock {\em Complexity and Approximation: {C}ombinatorial Optimization
  Problems and Their Approximability Properties}.
\newblock Springer-Verlag, Berlin, 1999.

\bibitem{BT10}
D.~Berend and T.~Tassa.
\newblock Improved bounds on bell numbers and on moments of sums of random
  variables.
\newblock {\em Probability and Mathematical Statistics}, 3(2), 2010.

\bibitem{Bod09}
H.L. Bodlaender.
\newblock Kernelization: {N}ew {U}pper and {L}ower {B}ound {T}echniques.
\newblock In {\em Proceedings of Parameterized and Exact Computation}, volume
  5917 of {\em LNCS}, pages 17--37. Springer-Verlag, 2009.

\bibitem{CC97}
L.~Cai and J.~Chen.
\newblock On fixed-parameter tractability and approximability of {NP}
  optimization problems.
\newblock {\em Journal of Computer and System Sciences}, 54:465--474, 1997.

\bibitem{CKP00}
A.~Caprara, H.~Kellerer, and U.~Pferschy.
\newblock The multiple subset sum problem.
\newblock {\em SIAM Journal of Optimization}, 11:308--319, 2000.

\bibitem{CKPP00}
A.~Caprara, H.~Kellerer, U.~Pferschy, and D.~Pisinger.
\newblock Approximation algorithms for knapsack problems with cardinality
  constraints.
\newblock {\em European Journal of Operational Research}, 123:333--345, 2000.

\bibitem{CT97}
M.~Cesati and L.~Trevisan.
\newblock On the efficiency of polynomial time approximation schemes.
\newblock {\em Inf. Process. Lett.}, 64(4):165--171, 1997.

\bibitem{CK00}
C.~Chekuri and S.~Khanna.
\newblock A {PTAS} for the multiple knapsack problem.
\newblock In {\em Proceedings of the {ACM-SIAM} Symposium on Discrete
  Algorithms}, pages 713--728. ACM-SIAM, 2000.

\bibitem{CHKX07}
J.~Chena, X.~Huang, I.A. Kanj, and G.~Xia.
\newblock Polynomial time approximation schemes and parameterized complexity.
\newblock {\em Discrete Applied Mathematics}, 155(2):180--193, 2007.

\bibitem{CT13}
G.~Cornuejols and R.~T\"ut\"unc\"u.
\newblock {\em Optimization {M}ethods in {F}inance}.
\newblock Cambridge University Press, New York, 2013.

\bibitem{CFKLMPPS15}
M.~Cygan, F.V. Fomin, L.~Kowalik, D.~Lokshtanov, D.~Marx, M.~Pilipczuk,
  M.~Pilipczuk, and S.~Saurabh.
\newblock {\em Parameterized Algorithms}.
\newblock Springer-Verlag, New York, 2015.

\bibitem{DF99}
R.G. Downey and M.R. Fellows.
\newblock {\em Parameterized Complexity}.
\newblock Springer-Verlag, New York, 1999.

\bibitem{DF13}
R.G. Downey and M.R. Fellows.
\newblock {\em Fundamentals of Parameterized Complexity}.
\newblock Springer-Verlag, New York, 2013.

\bibitem{EKMR15}
M.~Etscheid, S.~Kratsch, M.~Mnich, and R{\"o}glin.
\newblock Polynomial kernels for weighted problems.
\newblock In {\em Proceedings of Mathematical Foundations of Computer Science},
  volume 9235 of {\em LNCS}, pages 287--298. Springer-Verlag, 2015.

\bibitem{EKMR16}
M.~Etscheid, S.~Kratsch, M.~Mnich, and R{\"o}glin.
\newblock Polynomial kernels for weighted problems.
\newblock {\em Journal of Computer and System Sciences}, 2016.
\newblock to appear.

\bibitem{FGR10}
M.R. Fellows, S.~Gaspers, and F.A. Rosamond.
\newblock Parameterizing by the number of numbers.
\newblock In {\em Proceedings of the Symposium on Parameterized and Exact
  Computation}, volume 6478 of {\em Lecture Notes in Computer Science}, pages
  123--134. Springer-Verlag, 2010.

\bibitem{Fer05}
H.~Fernau.
\newblock {\em Parameterized Algorithmics: A Graph-Theoretic Approach}.
\newblock Habi\-li\-ta\-tions\-schrift, Universit{\"a}t T{\"u}bingen, Germany,
  2005.

\bibitem{FG06}
J.~Flum and M.~Grohe.
\newblock {\em Parameterized Complexity Theory}.
\newblock Springer-Verlag, Berlin, 2006.

\bibitem{FT87}
A.~Frank and E.~Tardos.
\newblock An application of simultaneous diophantine approximation in
  combinatorial optimization.
\newblock {\em Combinatorica}, 7(1):49--65, 1987.

\bibitem{Fre04}
A.~Fr{\'e}ville.
\newblock The multidimensional 0-1 knapsack problem: {A}n overview.
\newblock {\em European Journal of Operational Research}, 155:1--21, 2004.

\bibitem{GJ79}
M.R. Garey and D.S. Johnson.
\newblock {\em Computers and Intractability: A Guide to the Theory of
  NP-Completeness}.
\newblock W.H.~Freeman and Company, San Francisco, 1979.

\bibitem{GGRY15}
St.J. Goebbels, F.~Gurski, J.~Rethmann, and E.~Yilmaz.
\newblock Fixed-parameter tractability of change-making problems ({A}bstract).
\newblock International Conference on Operations Research (OR 2015), 2015.

\bibitem{GRY16}
F.~Gurski, J.~Rethmann, and E.~Yilmaz.
\newblock Capital budgeting problems: A parameterized point of view.
\newblock In {\em Operations Research Proceedings (OR 2014), Selected Papers},
  pages 205--211. Springer-Verlag, 2016.

\bibitem{GRY16a}
F.~Gurski, J.~Rethmann, and E.~Yilmaz.
\newblock Computing partitions with applications to capital budgeting problems.
\newblock In {\em Operations Research Proceedings (OR 2015), Selected Papers}.
  Springer-Verlag, 2016.
\newblock to appear.

\bibitem{Hro04a}
J.~Hromkovic.
\newblock {\em Algorithmics for Hard Problems: Introduction to Combinatorial
  Optimization, Randomization, Approximation, and Heuristics}.
\newblock Springer-Verlag, Berlin, 2004.

\bibitem{IK75}
O.H. Ibarra and C.E. Kim.
\newblock Fast approximation algorithms for the knapsack and sum of subset
  problem.
\newblock {\em Journal of the ACM}, 22(4):463--468, 1975.

\bibitem{IPZ01}
R.~Impagliazzo, R.~Paturi, and F.~Zane.
\newblock Which problems have strongly exponential complexity?
\newblock {\em Journal of Computer and System Sciences}, 63(4):512--530, 2001.

\bibitem{Jan12}
K.~Jansen.
\newblock A fast approximation scheme for the multiple knapsack problem.
\newblock In {\em Proceedings of the Conference on Current Trends in Theory and
  Practice of Computer Science}, volume 7147 of {\em LNCS}, pages 313--324.
  Springer-Verlag, 2012.

\bibitem{JLN10}
M.~J\"unger, T.M. Liebling, D.~Naddef, G.L. Nemhauser, W.R. Pulleyblank,
  G.~Reinelt, G.~Rinaldi, and L.A. Wolsey, editors.
\newblock {\em 50 Years of Integer Programming 1958-2008}.
\newblock Springer-Verlag, 2010.

\bibitem{Kan87}
R.~Kannan.
\newblock Minkowski's convex body theorem and integer programming.
\newblock {\em Mathematics of Operations Research}, 12:415--440, 1987.

\bibitem{KN05a}
S.-I. Kawano and S.-I. Nakano.
\newblock Constant time generation of set partitions.
\newblock {\em IEICE Trans. Fundam. Electron. Commun. Comput. Sci.},
  E88-A(4):930--934, 2005.

\bibitem{KPP10}
H.~Kellerer, U.~Pferschy, and D.~Pisinger.
\newblock {\em Knapsack Problems}.
\newblock Springer-Verlag, Berlin, 2010.

\bibitem{KS10}
A.~Kulik and H.~Shachnai.
\newblock There is no {EPTAS} for two-dimensional knapsack.
\newblock {\em Information Processing Letters}, 110(16):707--710, 2010.

\bibitem{Len83}
H.W. Lenstra.
\newblock Integer programming with a fixed number of variables.
\newblock {\em Mathematics of Operations Research}, 8:538--548, 1983.

\bibitem{LS55}
J.~Lorie and L.J. Savage.
\newblock Three problems in capital rationing.
\newblock {\em The Journal of Business}, 28:229--239, 1955.

\bibitem{MM57}
A.S. Manne and H.M. Markowitz.
\newblock On the solution of discrete programming problems.
\newblock {\em Econometrica}, 25:84--110, 1957.

\bibitem{MT90}
S.~Martello and P.~Toth.
\newblock {\em Knapsack Problems}.
\newblock John Wiley \& Sons, New York, 1990.

\bibitem{Mar08}
D.~Marx.
\newblock Parameterized complexity and approximation algorithms.
\newblock {\em The Computer Journal}, 51(1):60--78, 2008.

\bibitem{NLZ12}
J.~Nederlof, E.~J. van Leeuwen, and R.~van~der Zwaan.
\newblock Reducing a target interval to a few exact queries.
\newblock In {\em Proceedings of Mathematical Foundations of Computer Science},
  volume 7464 of {\em LNCS}, pages 718--727. Springer-Verlag, 2012.

\bibitem{Nie06}
R.~Niedermeier.
\newblock {\em Invitation to Fixed-Parameter Algorithms}.
\newblock Oxford University Press, New York, 2006.

\bibitem{Nie10}
R.~Niedermeier.
\newblock Reflections on multivariate algorithmics and problem
  parameterization.
\newblock In {\em Proceedings of the Annual Symposium of Theoretical Aspects of
  Computer Science}, volume~5 of {\em LIPIcs}, pages 17--32. Schloss Dagstuhl -
  Leibniz-Zentrum fuer Informatik, 2010.

\bibitem{PT99}
D.~Pisinger and P.~Toth.
\newblock Knapsack problems.
\newblock In {\em Handbook of Combinatorial Optimization}, volume~A, pages
  299--428. Kluwer Academic Publishers, 1999.

\bibitem{Var12}
M.J. Varnamkhasti.
\newblock Overview of the algorithms for solving the multidimensional knapsack
  problems.
\newblock {\em Advanced Studies in Biology}, 4:37--47, 2012.

\bibitem{Weg05}
I.~Wegener.
\newblock {\em Complexity Theory}.
\newblock Springer-Verlag, Berlin, 2005.

\bibitem{Wei66}
H.M. Weingartner.
\newblock Capital budgeting of interrelated projects: Survey and synthesis.
\newblock {\em Management Science}, 12(7):485--516, 1966.

\bibitem{WN67}
H.M. Weingartner and D.N. Ness.
\newblock Methods for the solution of the multidimensional 0/1 knapsack
  problem.
\newblock {\em Operations Research}, 15(1):83--103, 1967.

\end{thebibliography}
\end{document}